\documentclass[12pt]{amsart} 
\usepackage[papersize={210mm,297mm},left=20mm,right=50mm,top=20mm,bottom=40mm]{geometry}

\usepackage[colorlinks=true]{hyperref}
\hypersetup{linkcolor=slategray,urlcolor=blue,citecolor=blue}

\usepackage{color}
\definecolor{slategray}{RGB}{112,138,144}

\newcommand{\neu}[1]{#1}

\usepackage{mathrsfs}

\usepackage[colorlinks=true]{hyperref}
\hypersetup{linkcolor=slategray,urlcolor=slategray,citecolor=slategray}

\newtheorem{satz}{Theorem}

\newtheorem{lemma}[satz]{Lemma}
\newtheorem{koro}[satz]{Corollary}

\newcommand{\beq}{\begin{equation}}
\newcommand{\eeq}{\end{equation}}
\newcommand{\bea}{\begin{eqnarray}}
\newcommand{\eea}{\end{eqnarray}}
\newcommand{\beast}{\begin{eqnarray*}}
\newcommand{\eeast}{\end{eqnarray*}}
\newcommand{\beal}{\begin{align}}
\newcommand{\eeal}{\end{align}}
\newcommand{\bsp}{\begin{split}}
\newcommand{\esp}{\end{split}}

\newcommand{\bbbone}{{\mathchoice {\rm 1\mskip-4mu l} {\rm 1\mskip-4mu l}    {\rm 1\mskip-4.5mu l} {\rm 1\mskip-5mu l}}}
\newcommand{\bC}{{\mathbb C}}
\newcommand{\bI}{{\mathbb I}}
\newcommand{\bN}{{\mathbb N}}
\newcommand{\bR}{{\mathbb R}}
\newcommand{\bT}{{\mathbb T}}
\newcommand{\bZ}{{\mathbb Z}}

\newcommand{\de}{\delta}
\renewcommand{\th}{\theta} 

\newcommand{\vphi}{\varphi}
\newcommand{\veps}{\varepsilon}

\newcommand{\cA}{{\mathcal A}}
\newcommand{\cH}{{\mathcal H}}

\newcommand{\cV}{{\mathcal V}}

\newcommand{\scA}{{\mathscr A}}
\newcommand{\scB}{{\mathscr B}}
\newcommand{\scC}{{\mathscr C}}
\newcommand{\scF}{{\mathscr F}}
\newcommand{\scH}{{\mathscr H}}
\newcommand{\scI}{{\mathscr I}}
\newcommand{\scL}{{\mathscr L}}
\newcommand{\scM}{{\mathscr M}}
\newcommand{\scN}{{\mathscr N}}
\newcommand{\scT}{{\mathscr T}}
\newcommand{\scZ}{{\mathscr Z}}
\newcommand{\del}{\partial}

\newcommand{\gtoas}[1]{{\;\mathop{\longrightarrow}\limits_{#1}\;}}
\newcommand{\abs}[1]{{\left\vert #1 \right\vert}}
\newcommand{\norm}[1]{{\left\Vert #1 \right\Vert}}
\newcommand{\sfrac}[2]{{\textstyle \frac{#1}{#2}}}

\newcommand{\pli}{\prod\limits}

\newcommand{\E}{{\rm e}}
\newcommand{\I}{{\rm i}}
\newcommand{\dd}{{\rm d}}

\usepackage{graphicx}

\renewcommand\Re{\operatorname{Re}}

\newcommand{\grab}{\bar{\rm a}}
\newcommand{\gra}{{\rm a}}

\newcommand{\ul}[1]{\underline{#1}}
\newcommand{\Tr}{{\rm tr}}

\newcommand{\dirT}{\vec T}
\newcommand{\cdirT}{\vec{\scT}}

\newcommand{\Th}{\Theta}
\newcommand{\bTh}{\overline{\Theta}}
\newcommand{\bth}{\bar\theta}

\newcommand{\covC}{C}
\newcommand{\tV}{\tilde {\rm V}}

\newcommand{\teu}{t}
\newcommand{\intA}{{\mathcal{A}}}
\newcommand{\intV}{{\mathcal{V}}}
\newcommand{\Hchli}{\mbox{\sc h}_{\rm o}}
\newcommand{\Hach}{\mbox{\sc h}}
\newcommand{\Heps}{\Hach^{(\veps)}}
\newcommand{\sgm}{\;{\rm sgn}_-}
\newcommand{\gH}{\cH}
\newcommand{\ugH}{\tilde\cH}
\newcommand{\NbC}{\bC^{(N)}}
\newcommand{\siv}{s}
\newcommand{\iu}{\mathrm{i}}
\newcommand{\tnorm}[1]{|\!|\!|#1|\!|\!|}

\newcommand{\bracket}[2]{\langle#1\mid #2\rangle}
\renewcommand{\bbbone}{1}

\begin{document}

\title[Persistence of gaps for interacting fermionic systems]{Persistence of exponential decay and spectral gaps for interacting fermions}
\author{Wojciech de Roeck \& M.\ Salmhofer} 

\address{ 
Instituut voor Theoretische Fysica, K.~U.~Leuven, Celestijnenlaan~200D, 3001 Leuven, Belgium}

\email{wojciech.deroeck@kuleuven.be}

\address{ 
Institut f\" ur Theoretische Physik, Universit\" at Heidelberg,
Philosophenweg~19, 69120 Heidelberg, Germany}

\email{salmhofer@uni-heidelberg.de}

\date{\today}

\begin{abstract}
\noindent
We consider systems of weakly interacting fermions on a lattice. The corresponding free fermionic system is assumed to have a ground state separated by a gap from the rest of the spectrum. 
We prove that, if both the interaction and the free Hamiltonian are sums of sufficiently rapidly decaying terms, and if the interaction is sufficiently weak, then the interacting system has a spectral gap as well, uniformly in the lattice size. Our approach relies on convergent fermionic perturbation theory, thus providing an alternative method to the one used recently in \cite{H17}, and extending the result to include non-selfadjoint interaction terms. 
\end{abstract}

\maketitle

\section{Introduction}
In this paper we consider a general class of quantum many-body systems, namely systems of fermions with a Hamiltonian of the type $H = H_0 + H_I$, where $H_0$ describes independent particles (a quadratic Hamiltonian in second-quantized formalism) and $H_I$ is an interaction term which is assumed to be \emph{locally} small and to have a rapid spatial decay. For the purpose of illustration,
let us consider the simple, but relevant example of density-density interaction for spinless fermions, postponing the detailed and general theorems to a later stage:
$$
H_0=\sum_{x,x'} h_0(x,x') c^+_x c^-_{x'} -\mu \sum_x n_x,\qquad   H_I =  g \sum_{x, x'}v(x,x') n_x n_{x'}
$$
where $\mu \in \bR$, $g \in \bC$, $c_x^{\pm}$ are the fermionic creation/annihilation operators at site $x$, and $n_x=c_x^+c_x^-$ are number operators. The sums run over sites $x,x'$ of the finite regular lattice $\mathbb{Z}_L^d= (\mathbb{Z}/L\mathbb{Z})^d$ for some finite $L$, equipped with the graph distance $d(\cdot,\cdot)$.   Let us write $\Hchli$ for the bounded operator on $l^2(\mathbb{Z}_L^d)$ defined by the kernel $h_0$. A simplified version of our main result reads then 
 \begin{satz}\label{previewgaptheo}
 Assume that
 \begin{enumerate}
\item $h_0(x,x')=\overline{h_0(x',x)}$, i.e.\ $\Hchli$ and hence $H_0$ are self-adjoint.
\item  There are constants $c> 0$ and $C>0$ such that $|h_0(x,x')| \leq C \E^{-c d(x,x')}$ and $|v(x,x')| \leq C \E^{-c d(x,x')}$ for all $x$ and $x'$.
\item For some $\rho>0$, we have $\sigma(\Hchli-\mu) \cap [-\rho,\rho]=\emptyset$ (we denote by $\sigma(\cdot)$  the spectrum of an operator). 
\end{enumerate}
Then, for small enough $|g|$, $H$ has a simple eigenvalue $E_0$ such that $\Re E_0=\min (\Re \sigma(H))$ and
$$
\min \Re (\sigma(H) \setminus \{E_0\})-\Re E_0 \geq \rho -\mathcal{O}({|g|}^{\tfrac{1}{1+n}})
$$
where $n$ is the smallest integer {such that $n>d/2$}. The smallness condition on $g$ and the remainder term $\mathcal{O}({|g|}^{\tfrac{1}{n+1}})$ are uniform in the system size $L$. 
\end{satz}

Note that we do not need to assume translation invariance. 
Indeed, $\Hchli$ could be a Hamiltonian for a system in a weak random external potential (leaving the gap open), and no translation invariance is required for the interaction term either.  

Note also that $g$ may be complex. Our general result given in Theorem \ref{gaptheo} applies to general multibody interaction terms $H_I$ that need to be neither self-adjoint nor particle number conserving, only satisfy a smallness condition and a power-law decay of the kernels $v$ and $h_0$.  We decribe in Section  \ref{sec: decay estimates} how to specialize to Theorem \ref{previewgaptheo}. 

At first sight, it may seem surprising that the gap of the one-particle operator $\Hchli-\mu$ is $2\rho$, yet the lower bound for the gap of the many-body operator is at most $\rho$, not $2\rho$. Indeed, the lowest excitation (in the absence of interaction) is simply adding a fermion at energy $\rho$ or removing one at energy $-\rho$.  If we considered $H$ in a space with fixed fermion number, then indeed the lowest excitation would consist in raising a fermion from $-\rho$ to $\rho$ and we could expect a spectral gap close to $2\rho$.
  
 Within our approach, we do not obtain $n=0$ in the error estimate. This is in contrast to the case where $H_0$ can be represented as a frustration-free operator, treated in  \cite{H17,michalakis2013stability}.

The notion of a (volume-independent) spectral gap for a many-body system is crucial in quantum many-body theory. The rigorous theory got a considerable boost in the past decade, starting with the introduction of the quasi-adiabatic flow \cite{H05,automorphic} for lattice spin systems. The concept of this flow led to a definition of `gapped ground state quantum phases' that is now widely accepted: Two gapped ground states are in the same phase if their Hamiltonians can be continuously connected without closing the gap. 
Some prominent examples of rigorous results based on the quasi-adiabatic flow are (stability should be understood as referring to perturbations of the Hamiltonian that leave the gap open)
\begin{enumerate}
\item Stability of topological order \cite{bravyi2011short,bravyi2010topological}
\item  Stability of the area law of entanglement \cite{marien2016entanglement}
\item Quantization of Hall conductance \cite{HM,BBDF}, though this was also achieved with different methods in \cite{GMP}.
\end{enumerate}

As already indicated, the point of view in many of these results is to \emph{assume} that the gap remains open.  \emph{Proving} that there is a gap for a given Hamiltonian is in general hard. 
 It has been established for  weakly perturbed classical models, \neu{even with multiple ground states, \cite{borgs1996low,datta1996low,yarotsky2006ground}}, and more generally, for weak perturbations of frustration-free systems \cite{michalakis2013stability}. Furthermore, there is the \emph{martingale method}, see e.g.\ \cite{nachtergaele1996spectral} that can be used to prove a spectral gap for certain chains (one-dimensional systems).   Our paper adds a class of weakly interacting fermion systems to this list. 

We should point out now that most of what was discussed above (in particular the quasi-adiabatic flow) has only been explicitly established for spin systems up to now, and not for fermionic lattice systems. \neu{However, several recent papers have already  furnished some tools and results for the fermionic setting, e.g. \cite{katsura2017exact,monaco2017adiabatic}, and \cite{nachtergaele2017lieb} has started the task of systematically adapting the technical tools to fermionic systems. }

Let us now come back to our result and outline the paper. To make everything specific and reasonably concise, we assume that we are dealing with systems on a finite lattice. Continuum systems can be treated if either a short-distance regulator is included, or the spectrum of $\Hchli$ is sufficiently `widely spaced' at high energies. For a large part of our considerations, it suffices to assume that position space is a finite set $\Lambda$, and it does not matter whether it is a regular lattice or a more general graph. One essential point of our result is, of course, that the bounds are uniform in the number of elements of $\Lambda$, so that in the application, the limit of an infinite lattice or graph can be taken. 

Our proof relies on analyticity of fermionic perturbation theory. The truncated correlation functions are expressed as a convergent power series in the  interaction. In the form we use them here, the necessary (determinant) bounds were first proven in \cite{PeSa}, for the more special case of a translation-invariant system. The generalization of these bounds to general Hamiltonians only requires replacing Fourier transformation by the spectral theorem. The resulting bounds hold for any self-adjoint Hamiltonian $\Hchli$. The proof is given in Appendix \ref{detapp}.  

\subsection*{Acknowledgements}
We have benefitted from discussions with Marcello Porta at an early stage of this project. WDR also acknowledges the support of the Flemish Research Fund FWO under grant G076216N. 
\neu{The research of MS was supported by DFG Collaborative Research Centre {\sl SFB 1225 (ISOQUANT)}.}

\section{Setup and main result}

\subsection{The CAR algebra} 

Let $\Lambda$ be a finite set, $\scH = \ell^2 (\Lambda)$ and $\beta > 0$. Because $\scH$ is finite-dimensional, so is fermionic Fock space $\scF_\Lambda = \bigoplus_{n\ge 0} \bigwedge^n \scH$. Consequently, all linear operators on $\scF_\Lambda$ are bounded. We use the raising and lowering (creation and annihilation) operators $c^\pm_x$, satisfying the standard CAR
\beq\label{cliff}
\forall x,x' \in \Lambda \forall s,s' \in \{-1,1\}: \quad
c^s_x c^{-s'}_{x'} + c^{-s'}_{x'}c^s_x  = \delta_{x,x'} \delta_{s,s'}
\eeq
In the standard setup of \cite{BratRob}, the operators $c^\pm_x$ are associated to the orthonormal basis of $\scH$ given by the normalized eigenfunctions of the position operator, i.e.\ $c^\pm_x = c^\pm(f_x)$ where $f_x (y) = \delta_{x,y}$. 

{The unit anticommutator of (\ref{cliff}) is natural for a fermion system on a unit lattice. For a $d$-dimensional lattice with mesh size $a$, it would get scaled by a prefactor $a^{-d}$. This factor modifies the determinant bound $\delta$ (discussed below), and (as previous results), our bounds are not uniform in $a$ in the continuum limit $a\to 0$.}

\subsection{Operators and their normal-ordered form}

Let $\intA$ be a sequence of functions $\intA=(a_{\bar m, m})_{(\bar m,m) \in \bN_0^2}$, where for all $\bar m$ and $m$, 
\beq
a_{\bar m, m}: \Lambda^{\bar m} \times \Lambda^m \to \bC, 
\quad 
(\bar x_1, ... , \bar x_{\bar m}; x_1, ... x_m) 
\mapsto a_{\bar m, m} (\bar x_1, ... , \bar x_{\bar m}; x_1, ... x_m)\eeq
is antisymmetric: for any permutations $\pi$ on $\bN_{\bar m}$ and $\sigma$ on $\bN_m$,
\beq\label{antisy}
a_{\bar m, m} (\bar x_{\pi(1)} ,..., \bar x_{\pi(\bar m)}; 
x_{\sigma(1)}, ..., x_{\sigma(m)})
=
\veps_\pi\, \veps_\sigma\; 
a_{\bar m, m} (\bar x_1 , ..., \bar x_{\bar m}; 
x_{1}, ..., x_{m}) \; , 
\eeq
where $\veps_\pi$ denotes the sign of the permutation $\pi$.
To save space, we use the notations $\ul{x} = (x_1, \ldots, x_m)$
and $\ul{\bar x} = (\bar x_1, \ldots, \bar x_{\bar m})$.

Because $\scF=\scF_\Lambda$ is finite-dimensional, 
every linear operator $A$ on $\scF$ can be written as a polynomial in the $c^+$ and $c^-$, hence, using (\ref{cliff}), in the form
\beq\label{normA}
A
=
\scN (\intA)
=
\sum_{\bar m, m \ge 0} 
\sfrac{1}{\bar m!}
\sum_{\ul{\bar x} \in \Lambda^{\bar m}}
\sfrac{1}{ m!}\sum_{\ul{x} \in \Lambda^m}
a_{\bar m; m} (\ul{\bar x}, \ul{x})\; 
\pli_{\bar n=1}^{\bar m} c^+_{\bar x_{\bar n}}
\pli_{n=1}^{m} c^-_{x_{n}} \; .
\eeq
We shall call this the {\em normal ordered form} of $A$. 
For any $A$, the associated sequence $\intA$ is unique because of the antisymmetry condition (\ref{antisy}), so the normal ordered form is unique and can be regarded as a normal form. 

To make contact with notations used in spin systems, we briefly describe a different convention for writing $A$. 
Take any ordering of the finite set $\Lambda$. 
By nilpotency of the $c^\pm$, only injective sequences $\ul{x}$ and $\ul{\bar x}$ contribute. Denoting $M = \{ \bar x_1, \ldots, \bar x_{\bar m}\}$ and $N = \{ x_1, \ldots, x_{m}\}$, and setting 
\beq
w_{M,N} = a_{\bar m, m} (\bar x_{\pi(1)} ,..., \bar x_{\pi(\bar m)}; 
x_{\sigma(1)}, ..., x_{\sigma(m)})
\eeq
where the permutations $\pi$ and $\sigma$ are chosen such that 
$\bar x_{\pi(1)} ,..., \bar x_{\pi(\bar m)}$ is ascending in that order, and $x_{\sigma(1)}, ..., x_{\sigma(m)}$ is ascending as well, the expression for $A$ becomes
\beq\label{HIset}
A
=
\sum_{M,N \subset \Lambda} w_{M,N} 
\pli_{\alpha \in M} c^+_\alpha
\pli_{\beta \in N} c^-_\beta .
\eeq
We choose the convention (\ref{normA}) for writing $A$ because it does not require fixing any ordering of $\Lambda$, and will be more convenient when writing out the perturbation expansion later.

\subsection{The Hamiltonian}

The one-particle Hamiltonian $\Hchli$ is an operator on $\scH$, hence simply given by a matrix of  `hopping amplitudes' $h_0 (x,x')$ from $x'$ to $x$, which give $\Lambda$ the structure of a weighted directed graph (the directed edge $(x',x)$ has weight $h_0 (x,x')$). We always assume 
\beq
h_0 (x',x) = \overline{h_0 (x,x')}, 
\eeq
that is, $\Hchli$ is self-adjoint. 
We have motivated the matrix elements of $\Hchli$ as hopping amplitudes, {but $h_0 (x,x)$ can be nonzero. In particular,} $\Hchli$ may also include a chemical potential term $-\mu \delta_{x,x'}$ with $\mu \in \bR$. Moreover, the index $x$ can contain spin and other internal labels, besides the position.  A natural example would be to take $\Lambda= \mathbb{Z}^d_L \times \mathbb{S}$ where $\mathbb{Z}_L= \mathbb{Z}/{L\mathbb{Z}}$ can be viewed as discrete torus with $L$ sites and $\mathbb{S}=\{-s,\ldots, s-1, s\}$ corresponds to the $2S+1$ spin states of a spin $S$-system, with $2S \in \mathbb{Z}$. 

The second quantization of $\Hchli$,
\beq\label{H0def}
H_0 
=
(c^+, \Hchli \, c^-)_\Lambda
=
\sum_{x,x' \in \Lambda} 
h_0 (x,x') c^+_x c^-_{x'}
\eeq
describes independent fermions. The Hamiltonian for the system of interacting fermions is $H=H_0 + H_I$, where the interaction Hamiltonian $H_I$ of the fermions can contain general multibody interaction terms. In other words, it is a linear operator on $\scF_\Lambda$, given by a sequence of functions 
$\intV=(v_{\bar m, m})_{(\bar m,m) \in \bN_0^2}$ that are antisummetric in the sense of (\ref{antisy}).
We require $v_{0,0} = 0$. 
The interaction $\intV$ defines a second-quantized interaction Hamiltonian 
\beq\label{HIdef}
H_I 
=
\scN (\intV) \; .
\eeq
The condition $v_{0,0} =0$ restricts the sum over $\bar m $ and $m$ to pairs with $\bar m + m \ge 1$. This only removes the constant term, which drops out in all normalized expectation values. 

If $v_{\bar m, m} =0$ unless $\bar m + m$ is an even integer, the interaction $V$ is called even.
The condition that $H_I$ be self-adjoint puts an according one on
$V$.  

For simplicity of presentation, we will restrict to even $H_I$ in this paper, but this condition can easily be dropped, without any change in the results about exponential decay.
{\em Note that we do not require self-adjointness of $H_I$.}

\subsection{Norms}
The essential conditions on $\intV$, specified below, are (1)  `short-rangedness' in the form of a certain summability condition that is uniform in $|\Lambda|$, defined by a norm, and (2) `weakness' in the form of smallness of that norm. 
We define this norm for linear operators $A$ in terms of their normal ordered form, hence in terms of $\intA$, as follows. Let 
\beq
\abs{a_{\bar m; m}}_{1,\infty}^{(1)}
=
\sup_{\bar x_1} 
\sum_{\ul{\bar x}' \in \Lambda^{\bar m-1}}
\sum_{\ul{x} \in \Lambda^m}
\abs{a_{\bar m; m} (\bar x_1,\ul{\bar x}'; \ul{x})}
\eeq
and
\beq
\abs{a_{\bar m; m}}_{1,\infty}^{(2)}
=
\sup_{x_1} 
\sum_{\ul{\bar x} \in \Lambda^{\bar m}}
\sum_{\ul{x}' \in \Lambda^{m-1}}
\abs{a_{\bar m; m} (\ul{\bar x}; x_1,\ul{x}')}
\eeq
and set 
\beq
\abs{a_{\bar m; m}}_{1,\infty}
=
\max \left\{ 
\abs{a_{\bar m; m}}_{1,\infty}^{(1)}\;,
\;
\abs{a_{\bar m; m}}_{1,\infty}^{(2)}
\right\} \; .
\eeq
Moreover, let 
\beq
\abs{a_{\bar m; m}}_{1}
=
\sum_{\ul{\bar x}' \in \Lambda^{\bar m}}
\sum_{\ul{x} \in \Lambda^m}
\abs{a_{\bar m; m} (\ul{\bar x}'; \ul{x})}  \; .
\eeq
For $h> 0$ set 
\beq\label{normdef}
\norm{A}_h
=
\norm{\intA}_h
=
\sum_{\bar m, m \ge 0} 
\sfrac{1}{\bar m! m!}\;
\abs{a_{\bar m; m}}_{1,\infty}\; 
h^{\bar m + m} 
\eeq
and
\beq\label{tnormdef}
\tnorm{A}_h
=
\tnorm{\intA}_h
=
\sum_{\bar m, m \ge 0} 
\sfrac{1}{\bar m! m!}\;
\abs{a_{\bar m; m}}_{1}\; 
h^{\bar m + m} \; .
\eeq
Note that, for typical many-body operators like $A=\sum_x c^+_x c^-_x$, the norm $\tnorm{\intA}_h$ grows linearly in $|\Lambda|$, whereas the norm $\norm{A}_h$ stays bounded. The latter could be called a \emph{local norm} as it, roughly speaking, measures the size of the local terms. {For an operator $A$ given by a sequence $\cA$ of translation-invariant functions, $\tnorm{A}_h = |\Lambda|\; \norm{A}_h$.} 

\subsection{The truncated correlation function}
Let $\beta >0$. For a linear operator $A$ on $\scF_\Lambda$
set
\beq
\langle A \rangle_{\beta H}
=
\frac{1}{Z_{\beta H}} \Tr_\scF \left(
\E^{-\beta H} \; A
\right)
\eeq
The partition function is $Z_{\beta H} =\Tr_\scF \E^{-\beta H}$. 
If $H$ is self-adjoint, $Z_{\beta H}> 0$ follows from the spectral theorem. Because $\scF_\Lambda$ is finite-dimensional and $H_0$ is self-adjoint,  $Z_{\beta H_0} > 0$, and so $Z_{\beta H} \ne 0$ follows for small enough $V$ by continuity, since $\Lambda$ is finite. Under these conditions, the expectation value $\langle A \rangle_{\beta H}$ is well-defined for all $A$. 
The truncated correlation function of $A$ and $B$ is
\beq\label{truncAB}
\langle A(\tau) ; B \rangle_{\beta H}
=
\langle A(\tau)  B \rangle_{\beta H}
-
\langle A \rangle_{\beta H}\;
\langle  B \rangle_{\beta H}
\eeq
where 
\beq
A(\tau)
=
\E^{\tau H} \; A \E^{-\tau H} .
\eeq
Note that by cyclicity of the trace, $\langle A(\tau) \rangle_{\beta H}
= \langle A \rangle_{\beta H}$ for all $\tau$. The parameter $\tau \in [0,\beta) $ is a Euclidian time variable.   
{For finite $\Lambda$ and $\beta$, the truncated correlations are analytic functions of $V$. In the way we have just shown this, the radius of analyticity is obviously strongly dependent on the volume $|\Lambda|$, and on $\beta$. Our {results} imply, however, that these functions are analytic uniformly, that is, in a disk the radius of which is independent of $\Lambda$ and $\beta$.} 

\subsection{Exponential decay}
Let 
\beq\label{scCdef}
\scC (\tau, E) 
=
\bbbone_{\tau \le 0} f_\beta (E) \; \E^{-\tau E}
-
\bbbone_{\tau > 0} f_\beta (-E) \; \E^{-\tau E} 
\eeq
where $f_\beta (E) = (1+\E^{\beta E})^{-1}$ is the Fermi function and $\bbbone_{S} =1$ if statement $S$ is true, and $0$ otherwise. 
$\scC (\tau , \Hchli)$ is defined by the spectral theorem. 
The {\em fermionic covariance of $\Hchli$} is
\beq\label{fermcov}
C(\tau,x;\tau',x') 
=
\scC(\tau-\tau',\Hchli)_{x,x'}\; .
\eeq
The natural domain for the time-variable $\tau$ is the torus $\bR / \beta \bZ$, which we parametrize by $[0,\beta)$ with periodic boundary condition. Expressions such as $\tau-\tau'$ should hence be understood modulo $\beta$. 
In the following, let $d(\tau,\tau')$ be a metric on $\bR / \beta \bZ$. 
For $\rho \ge 0$, define $\alpha = \alpha_\rho$ as 
\beq
\alpha _\rho = \max \{ \alpha^+_\rho, \alpha^-_\rho\}
\eeq
where 
\beq
\alpha^+_\rho
=
\sup_{\tau, x}
\sum_{x'} \int_{-\beta}^\beta \dd\tau' \; 
\abs{C (\tau,x;\tau',x')} \E^{\rho d(\tau,\tau')} 
\eeq
and 
\beq
\alpha^-_\rho
=
\sup_{\tau, x}
\sum_{x'} \int_{-\beta}^\beta \dd\tau' \; 
\abs{C (\tau',x';\tau,x)} \E^{\rho d(\tau,\tau')} 
\eeq
{Although this is not made explicit in the notation, $\alpha_\rho$ is a function of the one-particle Hamiltonian $\Hchli$, and it also depends on which metric for $\tau$ is chosen.}

\begin{satz}\label{expodecaythm}
Let $H=H_0 + H_I$, with {$H_I=V$ given by an interaction $\cV$}, as in $(\ref{HIdef})$, and $A$ and $B$ be linear operators on $\scF$.  Recall the definition of $\norm{ \; \cdot \; }_h$ in $(\ref{normdef})$. Let $\delta =2$ and $\rho >0$. If $\alpha = \alpha _\rho$ is finite then for all $\tau \in [0,\beta)$, {the truncated correlation function of $A$ and $B$} is an analytic function of $V$ {on the ball $\alpha \norm{V}_{1+\delta} < 1$,  and it} decays exponentially in $\tau$:
\beq\label{ABdecay1}
\abs{\langle A(\tau) ; B \rangle_{\beta H}}
\le
\tnorm{A}_{1+\delta} \;  \norm{B}_{1+\delta} \; 
\frac{2 \alpha}{1 - \alpha \norm{V}_{1+\delta}}\;  \E^{-\rho\, d(0,\tau)}  \; .
\eeq
Moreover, 
\beq\label{ABdecay2}
\abs{\langle A(\tau) ; B \rangle_{\beta H}- \langle A(\tau) ; B \rangle_{\beta H_0}}
\le
\tnorm{A}_{1+\delta} \;  \norm{B}_{1+\delta} \; 
\frac{2 \alpha^2  \norm{V}_{1+\delta}}{1 - \alpha \norm{V}_{1+\delta}}\;  \E^{-\rho\, d(0,\tau)}  \; .
\eeq
Variants of the two above inequalities, obtained by replacing $\alpha$ with $\tilde \alpha = \alpha \delta^{-2}$ and $1+\delta$ with $2 \delta$ on the right hand side, also hold. 
Moreover, the product $\tnorm{A}_{1+\delta} \;  \norm{B}_{1+\delta}$ on the right side of the above inequalities can also be replaced by 
$\norm{A}_{1+\delta} \;  \tnorm{B}_{1+\delta}$.
\end{satz}

In brief, if the fermionic covariance of $\Hchli$ has an exponential decay, then for small enough \emph{local} norm of $H_I$, the truncated expectation of any two operators $A$ and $B$ decays exponentially as well. 

If $\alpha $ and $\delta$ are uniform in $\beta $ and $|\Lambda|$, so is the bound given by this theorem. {We have not explained yet where $\delta$ comes from. It is the determinant constant of the fermionic covariance of $\Hchli$. For the situation we consider, it turns out that for any self-adjoint $\Hchli$, one may take $\delta =2$. This is proven in the Appendix. When considering models on lattices of mesh size $\veps$, $\delta $ will depend on $\veps$.}

\subsection{Persistence of the gap}
\label{gapssec}

Recall that {because $\Lambda$ is finite, $H$ is a (possibly non-Hermitian) operator on a finite-dimensional space}, so it has a Jordan normal form
\beq
H= \sum_{j \ge 0} E_j \; P_j+D_j,\qquad D_j=P_jD_j= P_jD_jP_j
\eeq
where $E_0, E_1, \ldots$ are the eigenvalues of $H$, ordered such that $\Re E_0 {<}  \Re E_1 {<} \ldots$, $P_j$ the associated spectral projections, and the $D_j$ are nilpotents, i.e.\ there is a smallest $k_j \in \mathbb{N}$ such that $D_j^{k_j}=0$.  We may assume that $P_0 \ne 1$ to exclude the uninteresting case that $H$ is a multiple of the identity.
Without loss of generality, we may also assume that $E_0 = 0$, since we can always shift $H \to H - E_0$ without changing $\langle A(\tau) ; B \rangle_{\beta H}$.  Also, we will sometimes say that the gap is $c$, meaning that $\Re(E_1-E_0)=c$.
\begin{satz}\label{gaptheo}
Let $\rho >0$. 
If $\alpha = \alpha _\rho$ is finite, independent of $\beta$,  and $\alpha \norm{V}_{1+\delta} < 1$, then $E_0$ is a simple eigenvalue of $H = H_0 + V$ and  $\Re E_1- \Re E_0 \ge \rho$; i.e.\ {there is an `energy gap' at least $\rho$. If $\alpha_\rho $ is independent of $|\Lambda|$, then the gap is uniform in $|\Lambda|$.} 
\end{satz}
The proof uses a few simple lemmas. 
\begin{lemma}\label{lem: rep of gap}
Assume that $0$ is a simple eigenvalue of $H$ and that there is a positive gap, then 
$$
{\langle A(\tau) ; B\rangle_{\infty}} =  \Tr (P_0 A \E^{-\tau H} (1-P_0) B P_0)
$$
\end{lemma}
\begin{proof}
\neu{By the Jordan normal form, 
\beq
\E^{-\beta H}
=
\sum_{j \ge 0} \E^{-\beta E_j} \; P_j \E^{-\beta D_j} 
\gtoas{\beta\to \infty}
P_0 \; .
\eeq 
where we used $\Re E_1 \ge {\rho}$ (positive gap) and the fact that $\E^{-\beta D_j}$ is polynomial in $\beta$ (nilpotency). Similarly, $Z_{\beta H} \gtoas{\beta\to \infty}
\Tr (P_0) $.}  Therefore   
$$
\langle A(\tau) B\rangle_{\infty}
=
\sfrac{1}{\Tr (P_0)} \;
\Tr \left(
P_0 A \E^{-\tau H} B 
\right)   =     \Tr \left(
P_0 A P_0 B 
\right)    +   \Tr \left(
P_0 A \E^{-\tau H}  (1-P_0) B 
\right)  
$$ 
As $P_0$ is a rank $1$ projection, $\Tr \left(
P_0 A P_0 B 
\right)=  \langle A \rangle_{\infty}\langle B \rangle_{\infty}$ and Lemma \ref{lem: rep of gap} follows by the definition of the truncated correlation. 
\end{proof}

\begin{lemma}\label{spectrallemma}
Assume that $0$ is a simple eigenvalue and there is a positive gap.  If, for all $A,B  \in \scL (\scF)$, the bound 
\begin{equation} \label{lem: statement of bound}
|  \Tr (P_0 A \E^{-\tau H} (1-P_0)B ) | \leq  C(A,B) \E^{-\tau \rho}, \qquad \tau \geq 0,
\end{equation}
holds {with $C(A,B)$ independent of $\tau$},  then the gap is at least $\rho$, i.e.\
$$
\Re E_1 \geq \rho \; .
$$
\end{lemma}
\begin{proof}
Fix any index $j \neq 0$. Since there is at least one eigenvector associated to $E_j$, and $P_0$ has rank $1$, one can choose $A,B$ so that 
$\Tr \left[P_0 A \E^{-\tau H} (1-P_0) B \right] = \E^{-\tau E_j}$. Therefore, the hypothesis of the lemma implies that $\Re E_j \geq \rho$. 
\end{proof}

\begin{proof}[Proof of Theorem \ref{gaptheo}]
{For $g \in [0,1]$ let $H_g = H_0 + gV$, and let $E_j^g$ be the eigenvalues of $H_g$, again ordered such that $\Re E_j^g < \Re E_{j+1}^g $ for $j \ge 0$.} If at some value $g$, we knew that $E_0^g$ were simple, then 
Theorem \ref{expodecaythm} combined with the 
Lemmas \ref{lem: rep of gap} and \ref{spectrallemma} shows that the gap is at least $\rho$.  On the other hand, if $E_0^g$ is simple (and hence the gap is at least $\rho$), then by  perturbation theory the following bound is a sufficient condition for simplicity at $g+\Delta g$
\begin{equation}\label{eq: resolvent bound}
 \sup_{z: |z-E_0^g|=\rho/2} \norm{\frac{1}{z-H_g}}  < \frac{1}{ |\Delta g| \norm{ V}}
\end{equation} 
Hence, we have to bound a resolvent as in the left-hand side on the basis of spectral information. One has, in general, the bound (see e.g.\ Theorem 2.1.1 in \cite{gil2003operator})
\begin{equation} \label{eq: apriori resolvent}
\norm{\frac{1}{z-H}} \leq \sum_{k=0}^{d-1}  \frac{C(d,k)\norm{H}^k}{\mathrm{dist}(z,\sigma(H))^{k+1}},
\end{equation}
with $d$ the Hilbert space dimension and $C(d,k)$ combinatorial constants that do not depend on $H$. Note that for self-adjoint $H$, the bound is true already when dropping all terms with $k>0$. Back to our case: we have $d=2^{|\Lambda|}$ and we can easily find a bound $\norm{H+gV},\norm{V} \leq K$ such that $K$ depends on $\Lambda$  but not on $g$.   Therefore, we can find a ($\Lambda$-dependent) $\epsilon$ such that, for any $|\Delta g| \leq \epsilon$, the bound \eqref{eq: resolvent bound} is satisfied provided that $H_g$ has simple $E_0^g$ (and hence gap at least $\rho$).
Finally, we know that $E_0^0$ is simple by construction and by the above argument, we can prove simplicity for any $g \in [0,1]$ by iterating in steps of at most $\epsilon$ and using Theorem \ref{expodecaythm} at each iteration step.  This yields Theorem \ref{gaptheo}. 
\end{proof}

\subsection{Example: Integer Quantum Hall effect}

{This example is relevant for the quantum Hall effect and was discussed recently in \cite{BBDF}.
Let $\Lambda=\mathbb{T}^2_L = \bZ^2/L \bZ^2$ and let us label the lattice sites by  $x=(x_1,x_2)$ with $x_i =1,\ldots,L$. Moreover, let $\Hchli$ be self-adjoint and, for simplicity, finite-range, i.e.\ $h_0(x,x')=0$ whenever $|x-x'|_{\infty} >r$ for some fixed $r<L/2$.
The unperturbed Hamiltonian is given by (\ref{H0def}).} We do not specify the model further since that is not necessary for what follows, but a {possible} choice would be to take the Harper-Hofstadter model \cite{hofstadter1976energy}. A slight extension of the discussion below to non-cubic lattices would also allow to consider the Haldane model. 
We will thread the torus with magnetic fluxes, as follows.  We modify $h_0$ into 
\beq
h_0^{\phi}(x,x')=  h_0(x,x') \E^{\iu \phi(x,x')}
\eeq
with
\begin{equation} \label{eq: choice of phi}
\phi(x,x')= \phi_1 f(x_1,x'_1) +\phi_2 f(x_2,x'_2)
\end{equation}
where $\phi_{1,2}$ are both elements of $[0,2\pi)$ and 
\beq
f(y,y') =  \begin{cases} 1   &  \qquad \text{if} \quad  1 \leq y' < r,\quad  L-r+1 < y \leq L  \\   &  \qquad \text{if} \quad   1 \leq y < r,\quad L-r+1 <  y' \leq L   \\
0  & \qquad \text{otherwise}
 \end{cases}
\eeq
We can cast $\phi(x,x')$ as $\int_{\gamma_{x,x'}} d\mathbf{l} \cdot \mathbf{A}$, for a vector potential $\mathbf{A}$ defined on the edges of the lattice and $\gamma_{x,x'}$ a path between $x$ and $x'$.  The support of $\mathbf{A}$ is indicated in Figure \ref{fig: four squares}. We see that $\oint d \mathbf{l}\cdot \mathbf{A}=0$ for any closed loop that is contractible to a point. Hence there is no magnetic field piercing the torus. Yet, there are respective fluxes $(\phi_1,\phi_2)$  threading the torus. 
\begin{figure}
\centering
\includegraphics[width = 0.5\textwidth]{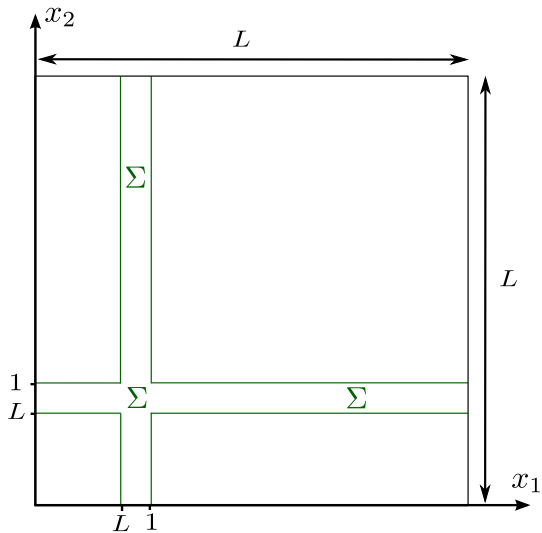}
\caption{The spatial support of the vector potential $\mathbf{A}$ is denoted by $\Sigma$. }
\label{fig: four squares}
\end{figure}
As it stands, the operators ${\Hchli}, \Hchli^{\phi}$ determined by the kernels $h_0$ and $h_{0}^{\phi}$ are not small (norm-)perturbations of each other, because  $|\phi|_{\infty}$ is not small unless the phases $\phi_1,\phi_2$ are small themselves. Let us consider a gauge transformation that remedies this. \neu{For a real-valued function $\nu$, let
$u_\nu$ the unitary operator that acts by multiplication by $ \E^{-\iu \nu(x)}$. Then 
\beq
u_\nu \Hchli^\phi u^{*}_\nu =   \Hchli^{\phi'},
\eeq
}with 
$$
\qquad \phi'(x,x')= \phi(x,x')+\nu(x)-\nu(x').
$$
We can now choose
\beq
\nu(x)=  \sum_{i=1,2} \phi_i (1-\tfrac{2x_i}{L}) \chi(1\leq x_i \leq L/2)
\eeq
This has the effect that, for the function $\phi$ chosen in \eqref{eq: choice of phi}, we obtain a transformed $\phi'$ that satisfies $|\phi'|_{\infty} \leq C/L$, \neu{and hence 
$ \norm{\Hchli- \Hchli^{\phi'}} \leq C/L$.  Therefore, gaps in the spectrum of $\Hchli$ remain open for $\Hchli^{\phi'}$, provided $L$ is large enough. }
The one-particle gauge transformation $u_\nu$ defines a many-body gauge transformation  
\beq
U_\nu = \E^{-\iu \sum_x \nu(x) n_x}
\eeq
(where $n_x =  c^+_x c_x^-$), and we have
\beq
U_\nu H^\phi_0 U^{*}_\nu =   H^{\phi'}_0,
\eeq
for the second-quantized operators $H^\phi_0,H^{\phi'}_0$ corresponding to $\Hchli^{\phi},\Hchli^{\phi'}$.  
The upshot is that, if $H_0$ has a spectral gap, then standard spectral perturbation theory for the one-body operators and the result of the present paper imply that this gap remains open when switching on a sufficiently weak interaction, uniformly for all fluxes $\phi$. The results of \cite{HM,BBDF} then imply that the Hall conductance is quantized. 

\section{Proof of Theorem \ref{expodecaythm}}
\subsection{The tree-determinant formula}

The proof of Theorem \ref{expodecaythm} uses an explicit formula for the truncated correlation function. It is similar to a quantum field theoretical Feynman graph summation where, however, a resummation over graphs is performed such that the sum reduces to tree graphs with associated determinants. The appearance of determinants is due to the fermionic antisymmetry; for bosons, one would have a permanent in its place. 

Let $R \subset \bN$ be finite. A set of ordered pairs 
\beq
\dirT
=
\{
(q,q') : q, q' \in R, q \ne q'
\}
\eeq
is a {\em directed tree} on $R$ if for any $q$ and $q'$, at most one of $(q,q')$ and $(q',q)$ is in $\dirT$, and if the corresponding set of unordered pairs
\beq
T
=
|\dirT|
=
\{
\{q,q'\} : (q,q') \in \dirT
\}
\eeq
is a tree with vertex set $R$. We denote the set of directed trees on $R$ by $\cdirT (R)$ and the set of trees on $R$ by $\scT(R)$. For the special case $R=\{ 1, \ldots , r\}$, we denote these sets by $\cdirT_r$ and $\scT_r$. 

If $\dirT \in \cdirT(R)$, then the sets 
\beq
\begin{split}
\Th_q 
&=
\{ (q',q): q' \in R, (q',q) \in \dirT \} 
\\
\bTh_q 
&=
\{ (q,q'): q' \in R, (q,q') \in \dirT \} 
\end{split}
\eeq
contain the ``ingoing'' and ``outgoing'' lines of $\dirT$ at vertex $q$. Let $\abs{\Th_q} = \th_q$ and $\abs{\bTh_q} = \bth_q$ be the graded incidence numbers.  
Then for all $q \in R$: $\th_q \ge 0$, $\bth_q \ge 0$, and 
$\th_q + \bth_q \ge 1$. 

The associated tree $T = |\dirT|$ has a sequence of incidence numbers $\ul{d} = (d_q (T))_{q \in R}$. They satisfy the tree relation 
$
\sum_{q\in R} d_q
=
2 (|R| -1). 
$
Evidently, if $T = |\dirT|$, then 
\beq\label{bththd}
\forall q \in R: \; \th_q+\bth_q = d_q \; .
\eeq

\begin{satz}\label{treedetthm}
Let $H=H_0 + \lambda H_I$. 
The truncated correlation function (\ref{truncAB}) is analytic in $\lambda$ for $|\lambda| < \lambda_0 (\beta, |\Lambda|, \norm{H_I})$. Its Taylor expansion around $\lambda =0$  can be rewritten as 
\beq\label{tauintformula}
\langle A(\tau) ; B \rangle_{\beta (H_0 + \lambda H_I)}
=
\sum_{p=0}^\infty \frac{\lambda^p}{p!}
\int_0^\beta \dd \tau_1 \ldots \int_0^\beta \dd \tau_p\;
G_{V,AB} (\tau_1, \ldots, \tau_p; \tau) \; ,
\eeq
where
\beq\label{treedetformula}
\begin{split}
& G_{V,AB} (\tau_1, \ldots, \tau_p; \tau)
\\
&=
\sum_{\dirT \in \cdirT_{p+2}}
\sum_{(\bar x_\ell, x_\ell)_{\ell \in \dirT}}
\pli_{\ell = (\bar q_\ell, q_\ell) \in \dirT}
\covC (\tau_{\bar q_\ell}, \bar x_\ell ; \tau_{q_\ell}, x_\ell) \;
\sum_{(\bar m_q,m_q)_{q\in \{1, \ldots, p+2\}}}
\sfrac{1}{\pli_{q=1}^{p+2} 
(\bar m_q - \bth_q)! (m_q -\th_q)!}
\\
&
\sum_{(\ul{\bar y}_q,\ul{y}_q)_{q\in \{1, \ldots, p+2\}}}
\pli_{q=1}^{p+2}
v_{\bar m_q,m_q}^{(q)} 
\left(
(\bar x_\ell)_{\ell \in \bTh_q}, \ul{\bar y}_q; (x_\ell)_{\ell \in \Th_q}, \ul{y}_q
\right)
\;
\sigma\; 
\left\langle 
\det{}_{\bar\nu,\nu} 
\left(
M \odot \covC
\right)
\right \rangle_+
\end{split}
\eeq
In this expression, $\tau_{p+1} = \tau$ and $\tau_{p+2} = 0$, 
$v_{\bar m_q,m_q}^{(q)} = v_{\bar m_q,m_q}$ for $q \in \{ 1, \ldots, p\}$, 
\beq
v_{\bar m_{p+1},m_{p+1}}^{(p+1)} = a_{\bar m_{p+1},m_{p+1}}
\qquad
v_{\bar m_{p+2},m_{p+2}}^{(p+2)} = b_{\bar m_{p+2},m_{p+2}} \; ,
\eeq
where $a$ and $b$ are the coefficient functions in the normal-ordered representation of $A$ and $B$ that is the direct analogue of $(\ref{HIdef})$. 
The $\bth_q$ and $\th_q$ are the graded incidence numbers of $\dirT$, and the summation over $\bar m_q$ and $m_q$ runs over the range $\bar m_q \ge \bth_q$ and $m_q \ge \th_q$. 
$\langle \; \cdot\;\rangle_+$ is the expectation over a $\dirT$-dependent probability measure of positive semidefinite $(p+2)\times (p+2)$ matrices $M$. The matrix  $\Gamma = M \odot \covC$ is defined as follows: let 
\beq\label{bnunudef}
\bar\nu = \sum_{q=1}^{p+2} (\bar m_q - \bth_q)
\quad
\mbox{ and }
\quad
\nu = \sum_{q=1}^{p+2} (m_q - \th_q)
\eeq
and let $\mu$ be the index $\mu=(q(\mu),r(\mu))$ with $q(\mu)=1,\ldots,p+2$ and $r(\mu) =1,\ldots, m_q-\theta_q$ and similarly for $\bar \mu=(q(\bar\mu),r(\bar\mu))$. Note that $\mu,\bar \mu$ index sets of cardinality $\nu,\bar\nu$. Let
then $\Gamma \in M_{\bar \nu, \nu} (\bC)$ with entries
\beq
\Gamma_{\bar\mu,\mu}
=
\left(
M 
\right)_{q(\bar\mu), q(\mu)} 
\;
\covC (\tau_{q(\bar \mu)}, \ul{\bar y}_{q(\bar \mu), r(\bar \mu)} ; \tau_{q(\mu)}, \ul{y}_{q(\mu), r(\mu)})
\eeq
where  $\ul{ y}_{q,r}$ is the $r$-th component of the vector $\ul{ y}_{q}$ (and similarly for $\ul{\bar y}_q$).

 The determinant $\det \Gamma$ is defined to be zero if $\bar \nu \ne \nu$.  

Finally, $\sigma $ is a combinatorial function of all summation variables which, however, takes only the values $-1$ and $1$. 
\end{satz}

\begin{figure}
\centering
\includegraphics[width = 0.75\textwidth]{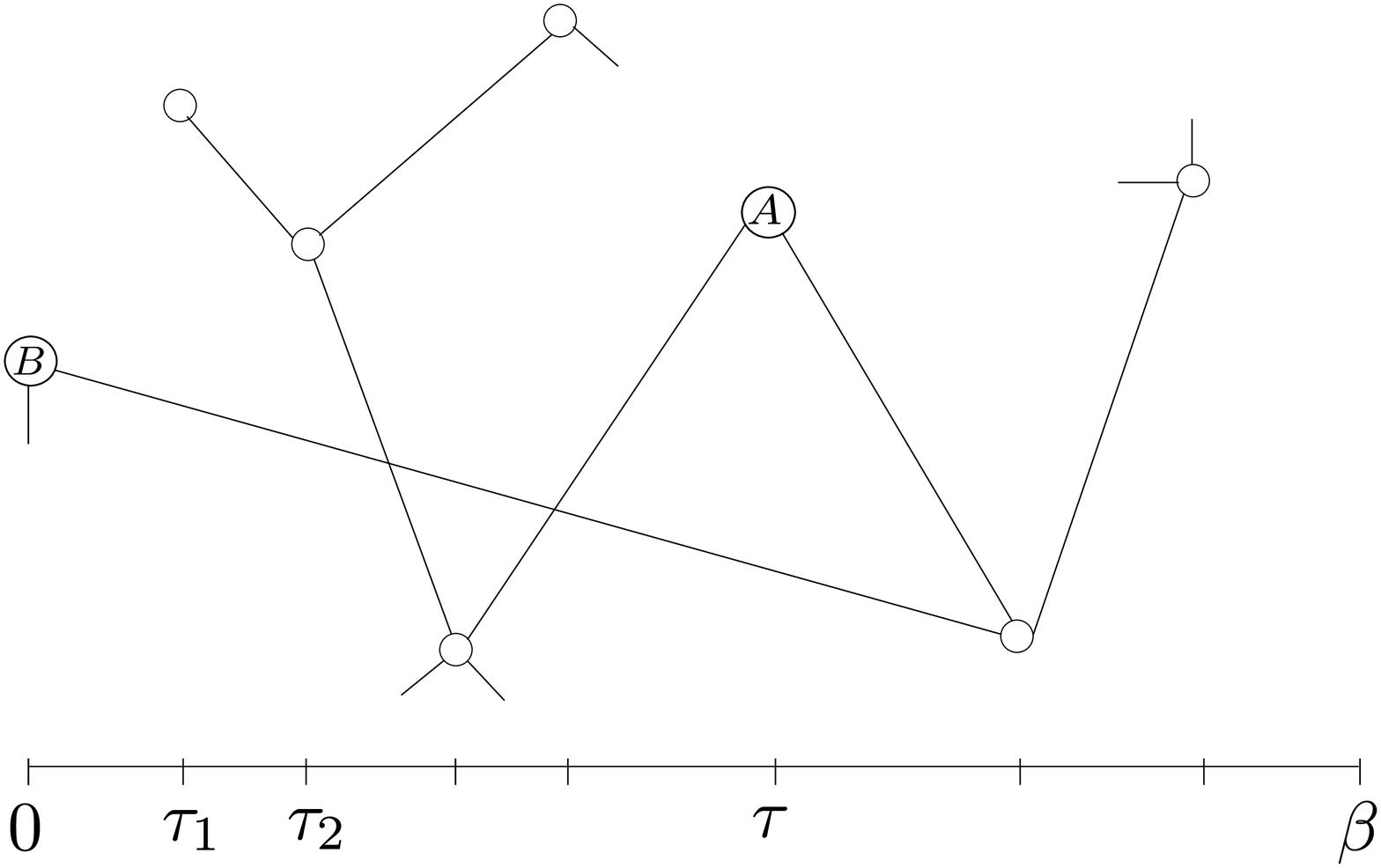}
\caption{An example of a $\dirT \in \cdirT_{p+2}$ with $p=6$ (we have drawn only the (undirected) tree $|T|$. To each vertex is associated a time $\tau_i$, with $\tau_{p+1}=\tau$ and $\tau_{p+2}=0$   and an interaction term $v_i$. The order of the interaction term (the number of field operators it contains) corresponds to the number of edges containing that vertex. The vertical direction schematically depicts the spatial location on $\Lambda$. {In the expansion, the locations of all vertices that are not labelled in this figure get summed over, corresponding to the sums over $(\bar x_\ell, x_\ell)_{\ell \in \dirT}$ in the tree-determinant-formula (\ref{treedetformula}), and the times $\tau_1$, $\tau_2, \ldots $ get integrated over, corresponding to the time integrals in (\ref{tauintformula})}.
}
\label{fig: tree1}
\end{figure}

\begin{figure}[!htb] 
\vspace{0.5cm}
\hspace{1cm}
 \begin{minipage}{.4\linewidth}
\centering
\includegraphics[width = \textwidth]{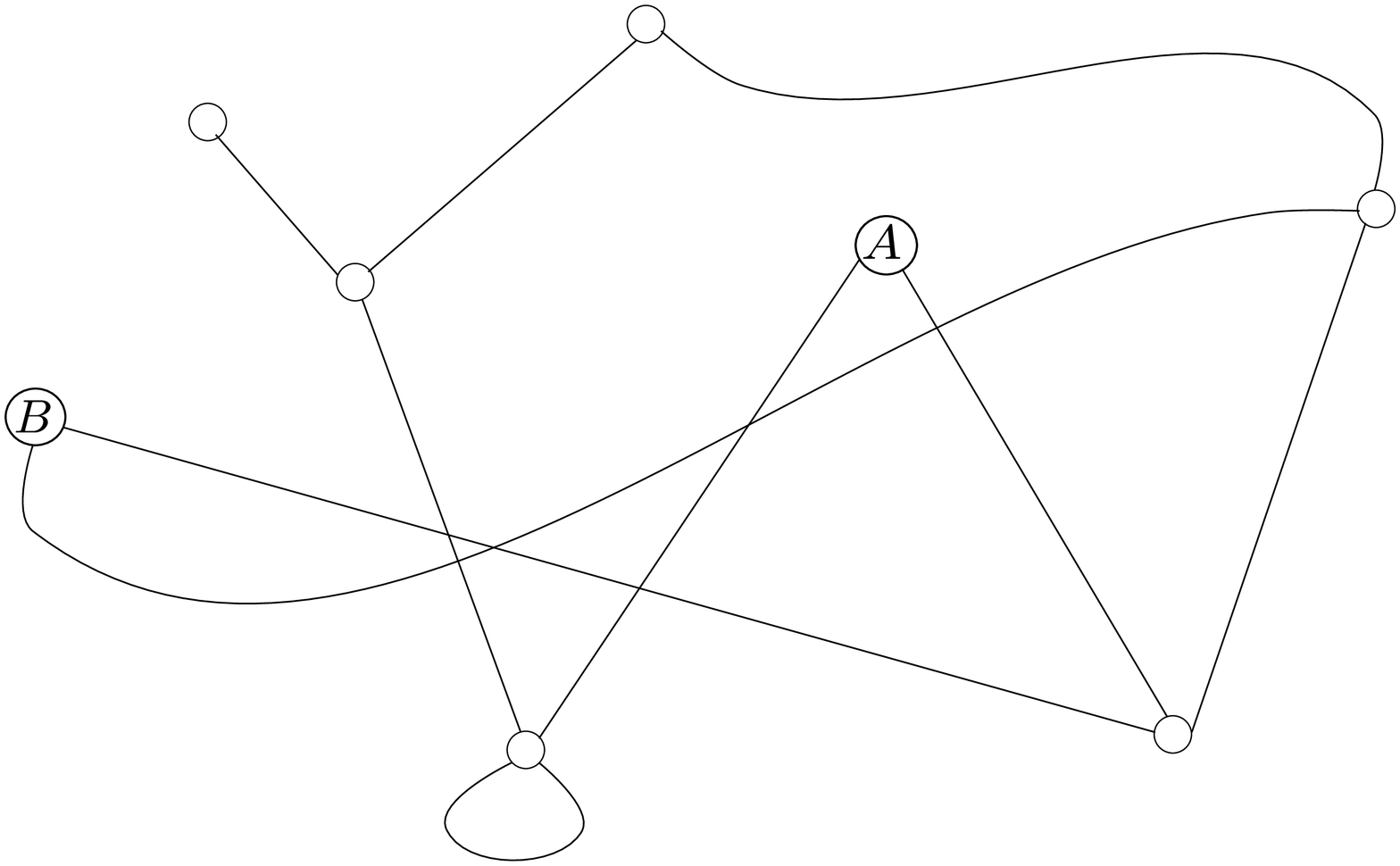}
 \end{minipage}
  \begin{minipage}{.4\linewidth}
\centering
\includegraphics[width = \textwidth]{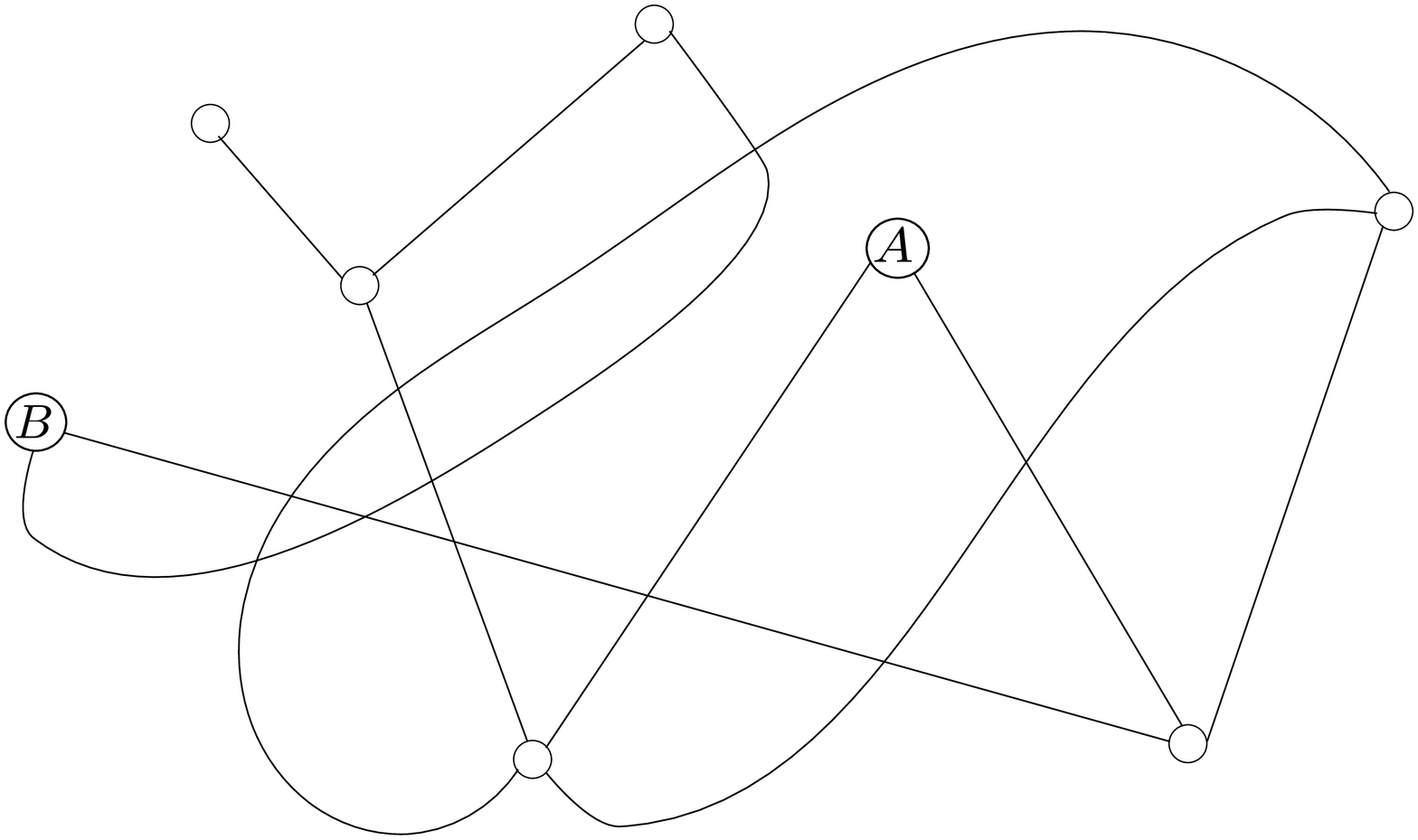}
 \end{minipage}
 \caption{Two examples of full (undirected) graphs associated to the above tree}\label{table rules}
\end{figure}

The proof of Theorem \ref{treedetthm} {is given} in Section \ref{peeksec}. It closely follows the proofs given in \cite{SaWi,Sa09}  (in fact, it is slightly simpler), so we will be brief about it and include only those details that make its structure clear.

{The stated analyticity in $\lambda$ (non-uniformly in $\beta$ and $|\Lambda|$) holds because for finite $\Lambda$, $\scF_\Lambda$ is finite-dimensional. This is briefly explained at the beginning of Subsection \ref{gen-fu}. The bounds given in Lemma \ref{Oplemma} then imply that the radius of convergence only depends on $\alpha_\rho$.}

Equation (\ref{treedetformula}) is not a short formula, but one should not be thrown off by the many summations involved:
the sums over $\bar m_q$, $m_q$, $(\bar x_\ell)_{\ell \in \bTh_q}$, $(x_\ell)_{\ell \in \Th_q}$, $\ul{\bar y}_q$ and $\ul{y}_q$ in (\ref{treedetformula}) are simply the summations appearing in the general representation of $H_I$ given by (\ref{HIdef}) and (\ref{normA}), only rearranged in terms of the directed tree. The important thing is that the product of covariances over tree lines connects the summations at each vertex and the $\tau_q$-integrations, so that a sufficient decay of $|\covC|$ implies exponential decay and uniform bounds. All the complications of the expectation value that go beyond the tree are in the determinant and in $\langle \; \cdot\;\rangle_+$.

This representation also has the standard Feynman graph interpretation \cite{msbook}, which is not required to do the proofs, but useful. Each factor $v^{(q)}_{\bar m_q,m_q}$ corresponds to a vertex with $\bar m_q$ outgoing and $m_q$ incoming lines. Each line in $\bTh_q$ joins $q$ to another vertex $q'$ by an outgoing line, and each line in $\Th_q$ joins some $q'$ to $q$ by an ingoing line. Because trees are connected, the Feynman graph is connected. The loops in the Feynman graphs that contribute to the truncated correlation function are created when the determinant is expanded in terms of permutations. 

The sign $\sigma$ will be irrelevant here because to get the bound in Theorem \ref{expodecaythm}, it suffices to take the absolute value inside all sums of (\ref{treedetformula}), after which it disappears. Similarly, details about $\langle \; \cdot\;\rangle_+$ will not be needed: since it is an average with respect to a probability measure, 
$\abs{\langle \det{}_{\bar\nu,\nu} M \odot \covC
\rangle_+}
\le 
\sup_M \abs{\det{}_{\bar\nu,\nu} M \odot \covC}$.
The essential point is that this determinant is bounded by const$^\nu$, uniformly in $M$.  

\subsection{Proof of exponential decay}
\label{proofsec}

We will prove the following bound for the term of order $p$:
\begin{lemma}\label{Oplemma}
Let $\rho > 0$ and 
\beq
\scI_p
=
\frac{1}{p!} \int_0^\beta \dd \tau_1 \ldots \int_0^\beta \dd \tau_p\;
\abs{G_{V,AB} (\tau_1, \ldots \tau_p; \tau)} \; \E^{\rho d(\tau,0)}
\eeq
Then 
\beq\label{jofrale}
\scI_p
\le
(2 \alpha_\rho)^{p+1} \; \norm{V}_{1+\de}^p \; 
\tnorm{A}_{1+\de} \; \norm{B}_{1+\de} \; .
\eeq
\end{lemma}

The estimates in Theorem \ref{expodecaythm} then follow by summation over $p$: summation over $p \ge 0$ gives (\ref{ABdecay1}). For (\ref{ABdecay2}), the summation starts at $p=1$, since the term of order zero is subtracted. 

The remaining part of this section contains the proof of Lemma \ref{Oplemma}. The procedure will be to take the absolute value inside all sums that are explicit in (\ref{treedetformula}), and to do the sum first over the variables $\bar x, x, \ul{\bar y}, \ul{y}$, then bound the sum over trees $\dirT$, and finally sum over $\bar m, m$. 

Let $\dirT \in \cdirT_{p+2}$ and $T = |\dirT|$ be the undirected tree associated to $\dirT$. There is a unique path $P$ from vertex $p+1$ to vertex $p+2$ over lines of $T$. If $r$ is the length of $P$ (which depends on $T$), we denote the vertices on $P$ by $\tilde q_0, \ldots , \tilde q_r$, with $\tilde q_0 = p+1$ and 
$\tilde q_r = p+2$, so that  $P=\{ \{\tilde q_{s-1}, \tilde q_s \} : s \in \{1, \ldots, r\}\}$. By the triangle inequality, 
\beq \label{eq: triangle times}
d(\tau,0)
\le
\sum_{s=1}^r d(\tau_{\tilde q_{s-1}}, \tau_{\tilde q_{s}})  \; .
\eeq
$T \setminus P$ is a forest of $r+1$ trees $\tilde T_s$ rooted at $\tilde q_s$. Let $\tV_s$ be the set of all vertices $q\ne \tilde q_s$ of $\tilde T_s$. ($\tV_s$ may be empty, if $\tilde T_s $ itself is empty.) 

Abbreviating $\covC_\ell = \covC (\tau_{\bar q_\ell}, \bar x_\ell ; \tau_{q_\ell}, x_\ell) $ and using \eqref{eq: triangle times}, we rearrange 
\beq
\E^{\rho d(\tau,0)} \;
\pli_{\ell = (q_\ell, q'_\ell) \in \dirT}
|\covC_\ell|  
\leq
\pli_{s=1}^r |\covC_{\ell_s}| \; \E^{\rho d (\tau_{\tilde q_{s-1}}, \tau_{\tilde q_s})} \;
\pli_{|\ell| \in \tilde T_s} |\covC_\ell| 
\eeq
where $|\ell |$ is the (undirected) edge corresponding to the directed edge $\ell$ and $\ell_s$ is the directed edge such that $|\ell_s|=\{ \tilde{q}_{s-1}, \tilde{q}_{s}\}$
and, similarly, 
\beq
\pli_{q=1}^{p+2}
\abs{v_{\bar m_q,m_q}^{(q)} }
=
\pli_{s=1}^r 
\abs{v_{\bar m_{\tilde q_s},m_{\tilde q_s}}^{(\tilde q_s)} }\;
\pli_{q \in \tV_s} 
\abs{v_{\bar m_q,m_q}^{(q)} }
\eeq
By Corollary \ref{cor: koro} (see Appendix), 
\beq
\abs{\left\langle 
\det{}_{\bar\nu,\nu} 
\left(
M \odot \covC
\right)
\right \rangle_+}
\le
\de^{\bar \nu + \nu} \; ,
\eeq
with $\bar\nu$ and $\nu$ given in (\ref{bnunudef}).
We now do the summations over $\bar x, x, \ul{\bar y}, \ul{y}$ and the $\tau_q$-integrals on each $\tilde T_s$, 
\beq
\Sigma_s (\tau_{\tilde q_s}, \tilde x_s)
=
\int \pli_{q \in \tV_s} \dd \tau_q \; \sum_{ {(\bar x_\ell, x_\ell)_{|\ell| \in \tilde T_s}\atop \tilde x_s \; {\rm fixed}}}
\sum_{(\ul{\bar y}_q, \ul{y}_q)_{q \in \tV_s}}
\pli_{|\ell| \in \tilde T_s} |\covC_\ell| 
\;
\pli_{q \in \tV_s} 
\abs{v_{\bar m_q,m_q}^{(q)} }
\eeq
The sum over $x$ excludes the one variable $x_\ell$ or $\bar x_\ell$ that also appears in  $v_{\bar m_{\tilde q_s},m_{\tilde q_s}}^{(\tilde q_s)}$, here denoted by $ \tilde x_s$. 
The sums in $\Sigma_s$ are done recursively, starting at the leaves of $\tilde T_s$ that are farthest away from the root $\tilde q_s$ (see, e.g.\ \cite{Br84,SaWi} for details).
Taking the supremum over this variable and over the time $\tau_{\tilde q_s}$ gives for all $s \in \{ 0, \ldots, r\}$
\beq
\sup\limits_{\tilde x_s, \tau_{\tilde q_s}}  
 \Sigma_s (\tau_{\tilde q_s}, \tilde x_s)
\le
\pli_{q \in \tV_s} \alpha_0 \;
|v_{\bar m_q,m_q}^{(q)} |_{1,\infty}
\eeq
Now the sum over all variables corresponding to lines $\ell $ on the path $P$ and the vertices $\tilde q_r, \ldots, \tilde q_0$, as well as the integrals over the times $\tau_{\tilde q_s}$ can be done similarly, giving a factor 
\beq
|v^{(\tilde q_0)}_{\bar m_{\tilde q_0}, m_{\tilde q_0}} |_{1} \;
\pli_{s=1}^r \alpha_\rho\; 
|v^{(\tilde q_s)}_{\bar m_{\tilde q_s}, m_{\tilde q_0}}|_{1,\infty} \; .
\eeq
Note that the variables in the vertex $\tilde q_0$, which corresponds to $A$, are all summed, so that at this place the norm is $|v^{(\tilde q_0)}_{\bar m_{\tilde q_0}, m_{\tilde q_0}} |_{1}$. 
Because $\rho > 0$, $\alpha_0 \le \alpha_\rho$. 
Thus
\beq
\scI_p
\le
\frac{\alpha_\rho^{p+1}}{p!}
\sum_{\dirT \in \cdirT_{p+2}}
\sum_{(\bar m_q, m_q)_{q}}
|v^{(\tilde q_0)}_{\bar m_{\tilde q_0}, m_{\tilde q_0}} |_{1}\;
\pli_{q=1 \atop q \ne \tilde q_0}^{p+2} 
|v_{\bar m_q,m_q}^{(q)} |_{1,\infty}\;
\sfrac{\de^{\bar m_q - \bth_q + m_q - \th_q}}{(\bar m_q - \bth_q)! (m_q -\th_q)!} \; .
\eeq
To restore the factorial in the definition of the norm, we rewrite 
\beq
\frac{1}{(m_q-\th_q)!}
=
\frac{1}{m_q!} \; {m_q \choose \th_q} \; \th_q!
\eeq
and similarly for the barred variables. 
We use (\ref{bththd}) to bound
\beq
\bth_q! \th_q! 
=
d_q! {d_q \choose \th_q}^{-1}
\le 
d_q!
\eeq
By Cayley's theorem, 
\beq
\frac{1}{p!} \pli_{q=1}^{p+2} d_q!
=
\frac{d_1 \ldots d_{p+2}}{|\{ T \in \scT_{p+2} : d_q (T) = d_q \forall q\}|} \; .
\eeq
By the arithmetic-geometric inequality and the constraint that the sum of incidence numbers is $2p+2$, $d_1 \ldots d_{p+2} \le 2^{p+1}$. Thus the dependence on the tree $\dirT$ has now reduced to a dependence on its graded incidence numbers 
$\bth_q$ and $\th_q$, and the sum can be bounded by a sum over 
sequences $\ul(d)=(d_1, \ldots d_{p+2}$ of incidence numbers, 
and a sum over $\bth_q$ and $\th_q$ satisfying (\ref{bththd}) at every $q$. The sum over $\ul{d}$ removes the constraint (\ref{bththd}), and we can now sum all $\th_q$ independently and get
\beq
\sum_{\th_q} {m_q \choose \th_q} \de^{m_q - \th_q} 
=
(1+\de)^{m_q} \; .
\eeq
All in all, we have
\beq
\scI_p 
\le 
(2\alpha_\rho)^{p+1}\; 
\sum_{(\bar m_q, m_q)_{q\in \{1, \ldots, p+2\}}}
\sfrac{\abs{v^{(\tilde q_0)}_{\bar m_{\tilde q_0}, m_{\tilde q_0}} }_{1}}{\bar m_{\tilde q_0}!m_{\tilde q_0}!} \;
\pli_{q=1 \atop q \ne \tilde q_0}^{p+2} 
\sfrac{\abs{v_{\bar m_q,m_q}^{(q)} }_{1,\infty}}{\bar m_q!m_q!} \; 
(1+\de)^{\bar m_q + m_q} 
\eeq
The right hand side equals that of (\ref{jofrale}), so Lemma \ref{Oplemma} is proven.  

Obviously, in the above proof, we could also have arranged the successive sum along the lines of $P$ so that $\tilde q_r$, which is associated to $B$, gets summed last, and in this way have obtained a bound with 
$\norm{A}_{1+\de} \; \tnorm{B}_{1+\de}$ instead of $\tnorm{A}_{1+\de} \; \norm{B}_{1+\de}$. This proves the last statement of Theorem \ref{expodecaythm}.

\section{Fermionic functional integrals}
\label{peeksec}

In this section, we briefly sketch how to obtain the tree-determinant expansion for the truncated correlation function. 
While one can, presumably, also get it from time-ordered expansions, it is much more convenient to use a functional integral representation, which also lends itself to more general purposes (e.g.\ multiscale analysis, which is useful in absence of a gap). Moreover, it will become clear that the field-theoretic generating function, from which the lengthy-looking formula (\ref{treedetformula}) is derived, really has a simple form and is amenable to standard techniques of mathematical quantum field theory and statistical mechanics. 

\subsection{Generating function}\label{gen-fu}
The truncated correlation function (\ref{truncAB}) can be obtained from the logarithm of the following generalized partition function. 
Let $r, s\in \bC$ and set 
\beq\label{Zab}
Z_{\beta H} (r,s)
=
\Tr \left(
\E^{-\beta H} \; 
(1+ r A(\tau)) \;  (1+sB)
\right)
\eeq
If this function is analytic in $r$ and $s$ for small enough $|r|$ and $|s|$, differentiation implies
\beq
\langle A(\tau) ; B \rangle_{\beta H}
=
\frac{\del^2}{\del r\del s} \ln Z_{\beta H} (r,s) \big\vert_{r=s=0} \; .
\eeq
Thus $\ln Z_{\beta H}$ can be regarded as the generating function for the truncated correlation of $A$ and $B$. 

Since $r$ and $s$ are just auxiliary parameters needed to derive a useful formula for $\langle A(\tau) ; B \rangle_{\beta H}$, it suffices to prove analyticity in a $\Lambda$- and $\beta$-dependent neighbourhood of $0$. 

{Recall that in the setting of Theorem \ref{treedetthm}, $H=H_0 + \lambda H_I$. Because $H_0$ is self-adjoint, $Z_{\beta H_0} > 0$. Because $\scF$ is finite-dimensional, $Z_{\beta H} (0,0)$ is an entire function of $\lambda$. If follows by continuity that there is $\lambda_0 = \lambda_0 (\beta, |\Lambda|, \norm{H_I}) > 0 $ such that for all $|\lambda| < 2 \lambda_0$, 
$Z_{\beta H} (0,0) \ne 0$. Since $Z_{\beta H} (r,s)$ is a polynomial in $r$ and $s$, it remains nonvanishing for small enough $r$ and $s$, hence $\log Z_{\beta H} (r,s)$ is analytic in $\lambda$, $r$, and $s$, in a small neighbourhood of $0$. }

In statistical mechanics, a polymer expansion is used to treat logarithms of partition functions, and to prove decay of correlations. In our setup, the algebra of observables is noncommutative. For this reason, we intermediately switch to a Grassmann integral representation of $Z_{\beta H} (r,s)$ and its logarithm. 
In this representation, the algebra involved is still noncommutative, but the even subalgebra is commutative, which is all we shall need to do bounds on the truncated function (\ref{truncAB}) by standard tree decay methods. 

Using Grassmann integrals has become standard in many-body theory, and our main theorem is an easy consequence of the estimates used in \cite{PeSa,Sa09,SaWi}, combined with a simple, but far-ranging, generalization of the determinant bound of \cite{PeSa} (proven in Appendix \ref{detapp}). 
Familiarity with this technique is, however, not required if one accepts the resulting formula for the tree expansion
of $\langle A(\tau) ; B \rangle_{\beta H}$, which is a particularly efficient representation of the perturbation expansion, in the sense that it makes fermionic sign cancellations easily visible, hence can be used to prove analyticity in the interaction. 

In the next subsection we include a few details about the functional integral, to make clear how the formula for the truncated correlation function that we use later on comes about. We include the present section also to make clear 

The rewriting in terms of Grassmann integrals done in the next subsection starts from a variant of the Lie-Trotter formula, and therefore the Grassmann algebras occurring here are finite-dimensional and that the `functional integrals' used here eally are linear functionals on these finite-dimensional spaces. This provides an additional proof that no `non-perturbatively small' remainder terms are missed in the sequel, so that it really suffices to bound coefficients of the perturbative expansion to show analyticity. Indeed, this rewriting and the proof of convergence in the limit $N \to \infty$ ($N$ is introduced below) does not require that the interaction is small or short-range; it holds in complete generality. 

\subsection{The functional integral representation}
It is well-known, see, e.g.\ the Appendix of \cite{msbook}, that the product of operators on Fock space, as well as the trace of any operator on Fock space, can be rewritten as a Grassmann integral by a simple algebraic procedure (summarized in Eqs.\ (42)-(44) of \cite{Sa09}), provided that these operators are given in normal ordered form. Indeed, from an algebraic point of view one can regard the combinatorial problem of an efficient calculation of traces involving $\E^{-\beta H}$ as equivalent to the one of writing it as $\E^{-\beta H} = \E^{-\beta H_0} D$, where $D$ is normal ordered, because then, the traces can easily be evaluated in terms of the quasifree density operator as determinants (for fermions) and permanents (for bosons). The Lie-Trotter product formula (\ref{Trottel}) is optimized so as to make this easy: if the interaction is given in normal-ordered form, each factor in this expression is either $\E^{-\veps H_0}$, or it is normal ordered. A motivation to do this via Grassmann integration is that the even subalgebra of a Grassmann algebra is commutative, so the logarithm of the partition function can be taken by standard methods of statistical mechanics, such as polymer expansions or by the interpolation technique used in \cite{SaWi}, which we also employ here. Since the algebra is done  in detail in Appendix B.5 of \cite{msbook} and in \cite{Sa09}, we only give the final result here. 

To distinguish the external time variable better, we denote it by $\teu $ instead of $\tau $ in the following. 
Let $0 < \teu < \beta$. We rewrite 
\beq
Z(r,s)
=
\Tr \left[
\E^{-(\beta-\teu) H}\;
(1+r A) \;
\E^{-\teu H} \;
(1+s B)
\right]
=
\lim\limits_{N \to \infty} Z_N (r,s)
\eeq
where $Z_N(r,s) = \Tr \; \Gamma_N (r,s)$ with 
\beq\label{Trottel}
\begin{split}
\Gamma_N (r,s)
=&
\left(
\E^{-\veps H_0}\;
(1-\veps H_I)
\right)^{N-2-k} \; 
\E^{-\veps H_0}\;
(1+r A)
\\
&
\left(
\E^{-\veps H_0}\;
(1-\veps H_I)
\right)^k
\E^{-\veps H_0}\;
(1+s B)
\end{split}
\eeq
Here $\veps = \frac{\beta}{N}$, and $k = k (\teu)$ is defined by 
\beq\label{kvont}
k = \lfloor \sfrac{\teu}{\beta} (N-2) \rfloor \; .
\eeq
Because $ \teu \in (0,\beta)$, both $k$ and $N-2-k$ go to infinity as $N \to \infty$. Thus, by the Lie product formula (in the form used in \cite{Sa09,msbook}) and by continuity of the trace on the finite-dimensional space $\scF$, $Z_N (r,s) \to Z(r,s)$ as $N \to \infty$. 

We denote the normal ordered form of $H_I$ by $\scH_I$, that of $A$ by $\scA$, and that of $B$ by $\scB$. 
Then our generating function $Z_N(r,s)$, given in (\ref{Zab}), can be rewritten as
\beq
\begin{split}
Z_N(r,s)
=
Z_0 \; 
\int \dd \mu_{\NbC} (\grab, \gra )
&\pli_{j=N-1}^{k+2} (1-\veps \scH_I (\grab_j, \gra_j)) \; 
(1 + r \scA (\grab_{k+1}, \gra_{k+1})) 
\\
&
\mkern10mu
\pli_{j=k}^{1} 
\mkern10mu
(1-\veps \scH_I (\grab_j, \gra_j)) \; 
(1 + s \scB (\grab_{0}, \gra_{0})) 
\end{split}
\eeq
where
\beq
Z_0 = Z_{\beta H_0} = \det (1 + \E^{\beta H_0}) \; ,
\eeq
and where we use the convention for product notation that $\pli_{j=k}^1 b_j = b_k b_{k-1} \ldots b_1$. Keeping the correct ordering is important if $H_I$ contains odd interaction terms. Here $\dd\mu_{\NbC}$ is the normalized Grassmann Gaussian measure on the Grassmann algebra with generators 
$(\grab_j (x), \gra_j(x))_{x \in \Lambda, j \in \{0, \ldots, N-1\}}$ and 
covariance
\beq
\NbC_{j,x; j',x'}
=
\scC (\tau^{(j)} - \tau^{(j')}, \Hchli)_{x,x'}
=
\covC(\tau^{(j)},x;\tau^{(j')},x') \; .
\eeq 
Thus, in the discretization implied by the choice of leaving $H_0$ in the exponent instead of writing $1-\veps H_I$ everywhere in (\ref{Trottel}), $\NbC$ is simply the fermionic covariance $\covC$ of (\ref{fermcov}), evaluated at the discrete times 
\beq\label{Riemerl}
\tau^{(j)} = j \veps = \frac{j}{N}\beta \; .
\eeq
The Grassmann Gaussian `measure' is not really a measure but rather a linear functional on the Grassmann algebra. As such it is uniquely defined by its action on monomials, which is
\beq
\int \dd \mu_{\NbC} (\grab, \gra)\;
\pli_{\bar k=\bar m}^{1} \grab_{j_{\bar k}, x_{\bar k}} \;
\pli_{k=1}^{m} \gra_{j'_k, x'_k} 
=
\delta_{\bar m, m}\;
\det \Gamma_m
\eeq
where $\Gamma_m $ is the $m \times m$ matrix with elements
\beq
\Gamma_{\bar k, k} 
=
\covC (\tau^{(j_{\bar k})}, x_{\bar k} ; \tau^{(j'_k)}, x'_k) \; .
\eeq
For the special case of a monomial of degree two, this gives the fermionic covariance: $\int \dd \mu_{\NbC} (\grab, \gra)\;\grab_{j, x} \;
\gra_{j', x'} = \covC(\tau^{(j)},x;\tau^{(j')},x')$. 

A simple application of the product inequality for the norm $\tnorm{\cdot}_1$ yields bounds for  Grassmann integrals that allow to exponentiate $1-\veps \scH_I  \to \E^{-\veps \scH_I}$ without changing the limit $N \to \infty$. 
For details, see Theorem B.14 of \cite{msbook} and equations (54)-(57) of \cite{Sa09}.

We now specialize to an even interaction $H_I$, in which case all factors commute and can be combined into a single exponential. 
Thus, in summary $Z_N (r,s) = Z_0 \; \scZ_N(r,s)$ with 
\beq\label{perfgoo}
\begin{split}
\scZ_N(r,s)
=
\int & \dd \mu_{\NbC} (\grab, \gra)\;
\E^{-\int_{\dot \bT} \scH_I(\grab(\tau),\gra(\tau) ) \dd \tau}\;
\\
&(1+ r \scA (\grab(\tilde \tau), \gra(\tilde \tau))) \; 
(1+ s \scB(\grab(0), \gra(0))
\end{split}
\eeq 
where we have labelled the Grassmann variables by $\tau^{(j)}$ instead of $j$, denoted $\tilde \tau = \tau^{(k(\teu))}$, with $k(\teu)$ given in (\ref{kvont}), $\bT = \{ j \veps : j \in 0, \ldots N-1\}$, and $\dot \bT = \bT \setminus \{ 0, \tilde \tau\}$, and used a continuum notation $\int \dd \tau$ for \neu{$\veps$ times} the sum over $j$, since the latter corresponds to the Riemann sum for this integral defined by the partition (\ref{Riemerl}). 
{Eq.\ (\ref{treedetformula}) can be derived straightforwardly from 
(\ref{perfgoo}). Instead of including this here, we briefly indicate how it can be brought into the standard form used in \cite{SaWi}. To this end, we note that the truncated correlation function is linear in $A$ and $B$, so higher orders in $r$ and $s$ do not matter. Thus we may replace $1+ r \scA (\grab(\tilde \tau), \gra(\tilde \tau))$  by $\exp \left[r \scA (\grab(\tilde \tau), \gra(\tilde \tau))\right]$, likewise for the term involving $\scB$, and obtain
\beq
\scZ_N(r,s)
=
\int \dd \mu_{\NbC} (\grab, \gra)\;
\E^{-\int_{\bT} \tilde\scH_I(\grab(\tau),\gra(\tau) ) \dd \tau}
\eeq
with
\beq
\begin{split}
\tilde\scH_I (\grab(\tau),\gra(\tau) )
=
\scH_I (\grab(\tau),\gra(\tau) )
&-
 r \delta (\tau,\tilde \tau) \scA (\grab(\tilde \tau), \gra(\tilde \tau)) \\
& -
 s \delta (\tau,0) \scB (\grab(\tilde \tau), \gra(\tilde \tau)) 
 \end{split}
\eeq
\neu{(and $\delta(\tau , \tau') = \veps^{-1} \delta_{j,j'}$ for $\tau=\tau^{(j)}$ and $\tau' = \tau^{(j')}$)}. 
This is in the form of a Grassmann Gaussian convolution, evaluated at zero external fields, that is, exactly the starting point of \cite{SaWi,PeSa}. Thus the results of these papers apply, in particular the interpolation formula, Theorem 3, of \cite{SaWi}, and the analyticity theorem, Theorem 4.5 of \cite{PeSa}, hold. The proof of analyticity is also included in our proof of exponential decay, given in Section \ref{proofsec}.
We note again that, although we have written that proof in the time continuum limit $N \to \infty$, the same estimates apply at finite $N$ because everything is streamlined such that the fermionic covariance at finite $N$ is the same function as the one for $N \to \infty$, just evaluated at discrete times, so the determinant bound is identical. For $N < \infty$, the decay bounds contain Riemann sums for the $\tau$-integrals, but these have the same estimates (up to $O(N^{-\alpha})$, $\alpha > 0$) as in the continuum, because the Riemann sums contain only uniformly bounded, continuous functions of $\tau$. Thus the convergence as $N \to \infty$ is uniform on compact sets in $\Lambda$ and $\beta$ if $\alpha_\rho$ is uniform, and the limiting function is therefore analytic in a disk of uniform radius.}

\appendix
\section{Determinant bound}\label{detapp}

In this appendix, we provide a determinant bound for fermionic covariances associated to arbitrary self-adjoint operators $\Hchli$ on general Hilbert spaces.
The argument will be a simple extension of that in \cite{PeSa}, the only change being that the Fourier transform used there is replaced by the spectral theorem. 

A different proof was given using noncommutative H\" older inequalities in \cite{BruPe}, and it gives a determinant constant that is a factor $2$ smaller. 
The proof we give here is elementary. 

The notion of determinant bound is given in Definition 1.2 of \cite{PeSa}. 

\begin{satz}\label{Hachsatz}
Let $\gH$ be a Hilbert space and $\Hach$ a self-adjoint operator on $\gH$. For $u,u' \in \gH$ and $\tau,\tau' \in [0,\beta]$ define
\beq\label{Cuup}
C_{\Hach} (\tau,u;\tau',u')
=
\langle u \mid 
\scC (\tau-\tau', \Hach) \; u'\rangle  \; .
\eeq
with $\scC$ given by (\ref{scCdef}). Then $C_{\Hach}$ has determinant bound 
\beq\label{depp}
\delta \le 2 \sqrt{\norm{u} \, \norm{u'}} \; .
\eeq
\end{satz}

\begin{proof}
By definition, 
\beq\label{pmdec}
C_{\Hach} (\tau,u;\tau',u')
=
\bbbone_{\tau \le \tau'} \; 
C_{\Hach}^- (\tau,u;\tau',u')
+
\bbbone_{\tau > \tau'} \; 
C_{\Hach}^+ (\tau,u;\tau',u')
\eeq
with
\beq
\begin{split}
C_{\Hach}^- (\tau,u;\tau',u')
&=\hphantom{-}
\bracket{u}{\mkern8mu f_\beta(\Hach) \mkern10mu\E^{- (\tau-\tau')\Hach} \, u'}
\\
C_{\Hach}^+ (\tau,u;\tau',u')
&=
- \bracket{u}{f_\beta(-\Hach) \; \E^{- (\tau-\tau')\Hach} \, u'} \;
\end{split}
\eeq
For $\veps > 0$ let $\Heps = \Hach + \veps \sgm(\Hach)$, where $\sgm (E) = -1$ for $E\le 0$ and $\sgm (E) = 1$ for $E > 0$. Then $\Heps$ has no spectrum in $(-\veps, \veps)$, and 
\beq
C^\pm_{H} (\tau,u;\tau',u')
=
\lim_{\veps\to 0} C^\pm_{\Heps} (\tau,u;\tau',u') .
\eeq
because at fixed $\beta$ and $t$, $E \mapsto \E^{-tE} f_\beta (E)$ is uniformly continuous on $\bR$. We shall show in Lemma \ref{kilogram} below that $C_{\Heps}^\pm (\tau,u;\tau',u')$ both have a Gram representation with Gram constant $\sqrt{\norm{u} \, \norm{u'}}$, which is independent of $\veps$. Eq.\ (\ref{depp}) then follows from Theorem 1.3 of \cite{PeSa} by taking the limit $\veps  \to 0$ (the factor $2$ is there because there are two summands in (\ref{pmdec})).
\end{proof}

\begin{koro}\label{cor: koro} The fermionic covariance to $\Hchli$ has determinant bound $\delta =2$, and
\beq\label{dedekind}
\abs{\det{}_{\bar\nu,\nu} M \odot \covC}
\le
\delta^{\bar\nu+\nu} \; .
\eeq
\end{koro}

\begin{proof}
In Theorem \ref{Hachsatz}, take $\Hach=\Hchli$, and note that the position space eigenfunctions have norm $1$, so that $\delta = 2$. Because $M$ is a positive matrix with diagonal elements equal to $1$, it has its own Gram representation with Gram constant $1$ \cite{SaWi}. Thus by definition of the determinant bound, (\ref{dedekind}) holds.
\end{proof}

If instead we use conventions for a lattice with mesh size $\veps$, then $\delta = 2 \veps^{-\frac{d}{2}}$: the determinant bound diverges as the lattice spacing goes to zero. This reflects the ultraviolet problem of many-body theory that arises in continuum systems. 

\begin{lemma}\label{kilogram}
Let $\gH$ be a Hilbert space, $\Hach$ a self-adjoint operator on $\gH$, and $u,u' \in \gH$. If for some $\veps > 0$ the spectrum of $\Hach $ excludes the interval $(-\veps, \veps)$, then $C^\pm_{\Hach} (\tau,u;\tau',u')$, defined in $(\ref{pmdec})$, both have Gram constant $\sqrt{\norm{u} \, \norm{u'}}$.
That is, there is a Hilbert space $\ugH$ and for all $\tau \in [0,\beta)$, $u\in \gH$, and $s\in \{ \pm\}$, there are vectors $\Phi^{s}_{\tau,u}\in\ugH$ and $\tilde \Phi^{s}_{\tau,u}\in\ugH$ with $\Vert\Phi^{s}_{\tau,u}\Vert \le \norm{u}$, $\Vert\tilde \Phi^{s}_{\tau,u}\Vert \le \norm{u}$  and
\beq\label{grammel}
C^s_{\Hach} (\tau,u;\tau',u')
=
\langle \Phi^{s}_{\tau,u} \mid \tilde\Phi^{s}_{\tau',u'}\rangle_{\ugH} .
\eeq
\end{lemma}

\begin{proof}
We give the proof for $C^+_{\Hach}$, hence may assume that $t=\tau-\tau'>0$. The proof for $C^-_{\Hach}$ is similar (with replacements as indicated in (61) of \cite{PeSa}).

For simplicity of presentation, assume first that $\gH$ is separable and that $\Hach$ has discrete spectrum. Denote the eigenvalues  of $\Hach$ by $E_n$, and the corresponding orthonormal basis of eigenvectors of $\gH$ by $\{\vphi_n:  n \in \bI\}$, where $\bI \subset \bZ$.  (Higher multiplicity of eigenvalues is allowed.) 
Let $\bI_\pm = \{ n \in \bI: \pm E_n > 0\}$. 
By the spectral theorem
\beq
C^+_{\Hach} (\tau,u;\tau',u')
=
C^+_{\Hach,+} (\tau,u;\tau',u')
+
C^+_{\Hach,-} (\tau,u;\tau',u')
\eeq
with
\beq\label{sumn}
C^+_{\Hach,\pm} (\tau,u;\tau',u')
=
\sum_{n\in \bI_\pm}
\langle u \mid \vphi_n \rangle \; 
\E^{-(\tau-\tau') E_n} f_\beta ( - E_n) \; 
\langle \vphi_n  \mid u'\rangle
\eeq
(and the sum is absolutely convergent for all $u,u' \in \gH$). 
Because $t=\tau-\tau'>0$,
\beq\label{sint}
\E^{-t E_n}
=
\frac{E_n}{\pi} 
\int_\bR
\frac{\E^{\I \siv t}}{\siv^2 + E_n^2} \; \dd \siv\; 
\eeq
holds for $E_n > 0$ and this integral is absolutely convergent, 
so
\beq
\begin{split}
C^+_{\Hach,+} (\tau,u;\tau',u')
&=
\sum_{n\in \bI_+}
\int_\bR \frac{\dd \siv}{\pi}\; 
\langle u \mid \vphi_n \rangle \E^{\I \siv \tau}\; 
\sfrac{E_n\; f_\beta ( - E_n)}{\pi (\siv^2 + E_n^2)} \;
\langle \vphi_n  \mid u'\rangle \E^{- \I \siv \tau'}
\\
&=
\langle \phi_{\tau,u} \mid \phi_{\tau',u'}\rangle_{\ugH_+}
\end{split}
\eeq
with $\ugH_+ = \ell^2 (\bI_+) \otimes L^2(\bR)$ and 
\beq
\phi_{\tau,u} (n,\siv)
=
\left(
\frac{E_n\; f_\beta ( - E_n)}{\siv^2 + E_n^2}
\right)^{\frac12}
\langle \vphi_n  \mid u\rangle \E^{- \I \siv \tau} \;.
\eeq
Moreover
\beq
\begin{split}
\norm{\phi_{\tau,u} }_{\ugH_+}^2
&=
\sum_{n\in \bI_+}
\int_\bR \frac{\dd \siv}{\pi}\; 
\frac{E_n\; f_\beta ( - E_n)}{\siv^2 + E_n^2}
\langle u \mid \phi_n \rangle \; 
\langle \phi_n  \mid u\rangle
\\
&=
\sum_{n\in \bI_+}
\langle u \mid \phi_n \rangle \; 
f_\beta ( - E_n)\;
\langle \phi_n  \mid u\rangle
\\
&=
\langle u \mid
\bbbone_{H > 0}\; f_\beta ( - H) \; u \rangle \; .
\end{split}
\eeq
If $E_n < 0$, we rewrite 
\beq
\E^{-t E_n } f_\beta (-E_n)
=
\E^{(\beta - t) E_n} f_\beta (E_n)
=
\E^{-(\beta - t) |E_n|} f_\beta (- |E_n|)
\eeq
and use (\ref{sint}) with $t$ replaced by $\beta - t \ge 0$ and $E_n$ replaced with $|E_n| >0$. Proceeding as before, we get
\beq
C^+_{\Hach,-} (\tau,u;\tau',u')
=
\langle \psi_{\tau,u,+} \mid \psi_{\tau',u',-}\rangle_{\ugH_-}
\eeq
with $\ugH_- = \ell^2 (\bI_-) \otimes L^2(\bR)$ and 
\beq
\psi_{\tau,u,\pm} (n,\siv)
=
\left(
\frac{|E_n|\; f_\beta ( - |E_n|)}{\siv^2 + |E_n|^2}
\right)^{\frac12}
\langle \vphi_n  \mid u\rangle \E^{- \I \siv (\tau\mp\frac{\beta}{2})}
\eeq
satisfying
\beq
\norm{\psi_{\tau,u,\pm} }_{\ugH_-}^2
\le
\langle u \mid
\bbbone_{\Hach < 0}\; f_\beta ( \Hach) \; u \rangle \; .
\eeq
Set $\ugH = \ugH_+ \oplus \ugH_-$ and 
$\Phi^{+}_{\tau,u} = \phi_{\tau,u} \oplus \psi_{\tau,u,+}$,
$\tilde \Phi^{+}_{\tau,u} = \phi_{\tau,u} \oplus \psi_{\tau,u,-}$, 
Then (\ref{grammel}) holds, and 
\beq
\begin{split}
\Vert\Phi^{+}_{\tau,u}\Vert^2
&=
\norm{\phi_{\tau,u}}_{\ugH_+}^2 
+
\norm{\psi_{\tau,u,+}}_{\ugH_-}^2
\\
&=
\bracket{u}{ \left[\bbbone_{H > 0}\; f_\beta ( - \Hach)
+ \bbbone_{\Hach < 0}\; f_\beta ( \Hach)
\right] u}
\\
&\le
\bracket{u}{ \left[\bbbone_{H > 0}
+ \bbbone_{\Hach < 0}
\right] u}
\\
&= \norm{u}^2 \; .
\end{split}
\eeq
The estimate for $\tilde\Phi^+_{\tau,u}$ is the same.

The proof for general self-adjoint Hamiltonians similar, except that we have to use the general form of the spectral theorem: there is a measure space $(\scM, \Sigma, \mu)$, a unitary map $U: \gH \to L^2 (\mu)$, and a measurable function $\eta: \scM \to \bR$ such that 
\beq
\langle u , \E^{-t H} f_\beta ( - H) u' \rangle
=
\int \dd \mu (\xi) \; 
\overline{ (Uu)(\xi)} \; 
\E^{-t \eta(\xi)} \phi_\beta ( - \eta(\xi))\; 
(Uu')(\xi) .
\eeq 
The rest of the argument is literally the same as before, replacing the $\ell^2$ spaces with the corresponding subspaces of $L^2(\mu)$ and the sums with integrals. 
\end{proof}

\section{Decay estimates: Finiteness of $\alpha_\rho$}
\label{sec: decay estimates}

The point here is that we do not only need to make explicit the exponential decay in $\tau$, which will naturally follow from a spectral assumption on $\Hchli$  but also use the summability in $x$ to get a bound that is indeed uniform in $|\Lambda|$. 
Let us write 
\beq
k(\zeta)= \sup_x\sum_{y} (1+d(x,y))^{-\zeta}, \qquad \zeta >0
\eeq
A possible bound on the decay constant $\alpha_\rho$ is given by the following theorem which is, at heart, a Combes-Thomas estimate (see also \cite{aza2017large}).  \neu{Choose the metric $d(\tau,\tau')$ in the definition of $\alpha_\rho$ to be given by the absolute value of  $\tau-\tau'$ modulo $\beta$.}
\begin{satz}\label{thm: decay from gap}
Choose \neu{$0< \epsilon, \rho < \norm{\Hchli}$} such that $\Hchli$ has no spectrum in the interval $[-\rho-\epsilon,\rho+\epsilon]$. Assume that
\beq
|h(x,x')| \leq K_{\nu} (1+d(x,x'))^{-\nu}, 
\eeq
for some  $\nu>0$ and  $K_{\nu} <\infty$. Let $n \in \mathbb{N}$ with $1 \leq n<\nu$. 
Then, ,
\begin{equation}  \label{eq: bound alpha}
\alpha_\rho  \leq   C  \sqrt{k(2n)} k(\nu-n)^n  \epsilon^{-n-1}
\end{equation}
where $C$ is a constant depending only on $n, K_{\nu} $ and \neu{$\norm{\Hchli}$}.
\end{satz}

\begin{proof}[Proof of Theorem \ref{thm: decay from gap}]
We write $C$ for various constants depending only on $n, K_{\nu},\norm{\Hchli}$. 
We prove that, uniformly in $\beta$, and $\tau\leq 0$,
\begin{equation}  \label{eq: proof reduces}
 \sup_x\sum_{y}  |(x, \chi(\Hchli \leq 0){f_\beta(\Hchli)e^{-\tau \Hchli}} y ) |  \leq   C \E^{(\rho+\epsilon/2)\tau}   \sqrt{k(2n)} k(\nu-n)^n \epsilon^{-n}
 \end{equation}
which takes care of the first term in \eqref{scCdef}, restricted to the negative part of $\Hchli$. \neu{The other $3$ terms} are dealt with analogously, and the bound  \eqref{eq: bound alpha} follows. 

Let $\Gamma$ be a contour that encircles the negative spectrum of $\Hchli$ counterclockwise (and no other spectrum). By the spectral theorem
\beq
F(\Hchli)\equiv \chi(\Hchli \leq 0){f_\beta(\Hchli)e^{-\tau \Hchli}}  =  \frac{1}{2\pi \iu}\oint_{\Gamma} d z f_\beta(z) e^{-z \tau} \frac{1}{z-\Hchli}
\eeq
so that 
\begin{equation}\label{eq: to bound}
\sum_y |(x, F(\Hchli)y)| \leq  \frac{1}{2\pi} \oint_{\Gamma} |d z |  | f_\beta(z) e^{-z \tau}|  \sum_y |(x, \frac{1}{z-\Hchli} y)|
\end{equation}
To handle the $y$-sum, we introduce
\beq
\Hchli^{x,\kappa} =   \E^{\iu \kappa d(\cdot,x)} \Hchli  \E^{-\iu \kappa d(\cdot,x)},\qquad \kappa \in \mathbb{R}
\eeq 
where, for any $x\in \Lambda$, the unitary operator $\E^{-\iu \kappa d(\cdot,x)}$ acts by multiplication on $\Lambda$. 
For any $x\neq y$, we have then
\begin{equation}
\label{eq: rep derivative}
(x,\frac{1}{z-\Hchli}y) = (d(x,y))^{-n} (-i\partial_\kappa)^n  (x,\frac{1}{z-\Hchli^{x,\kappa}} y)\Big|_{\kappa=0}
\end{equation}
and using Cauchy-Schwarz, we get
\begin{equation}
\label{eq: bound propa}
\sum_{y\neq x} |(x, \frac{1}{z-\Hchli} y)| \leq  \sqrt{k(2n)} \norm{ \partial_\kappa^n  \frac{1}{z-\Hchli^{x,\kappa}} \Big|_{\kappa=0}} .
\end{equation}
We now control the operator norm on the right.  By repeatedly using the resolvent identity (for bounded operators and $z$ in the resolvent sets of $A,A+B$),
\beq
\frac{1}{z-(A+B)} - \frac{1}{z-A} =  \frac{1}{z-(A+B)} B \frac{1}{z-A},
\eeq
we find
\begin{equation}
\label{eq: bound norms}
\norm{ \partial_\kappa^n  \frac{1}{z-\Hchli^{x,\kappa}}}  \leq C  \norm{\frac{1}{z-\Hchli^{x,\kappa}}}^{n+1}
\sup_{(p_i):\sum_i p_i=n}\prod_i\norm{\partial^{p_i}_\kappa \Hchli^{x,\kappa}}
\end{equation}
Since $\Hchli^{y,\kappa}$ is unitarily equivalent to $\Hchli$ we can bound 
\beq 
\norm{\frac{1}{z-\Hchli^{x,\kappa}}}  \leq   \frac{1}{d(z,\sigma(\Hchli))}. 
\eeq
To deal with the norm of $\kappa$-derivatives on the right-hand side of \eqref{eq: bound norms}, we recall the following bound (Schur's test) for an operator $\Hach$ with kernel $h$;
\beq
\norm{\Hach}^2 \leq \big(\sup_{x}\sum_{x'}| h(x,x')| \big) \big(  \sup_{x'}\sum_{x}| h(x,x')| \big).
\eeq
Moreover, the kernel $h^{(p)}$ associated to the operator $\Hach=\partial^p_\kappa \Hchli^{y,\kappa}$ is bounded by
\beq
| h^{(p)}(x,x') |  \leq d(x,x')^p |h_0(x,x')|
\eeq
as follows from the triangle inequality.
Hence, uniformly in $y$, 
\beq
\sup_{(p_i):\sum_i p_i=n}\prod_i\norm{\partial^{p_i}_\kappa \Hchli^{y,\kappa}}   \leq   C k(\nu-n)^n
\eeq
so that the left-hand side of \eqref{eq: to bound} is bounded by.
\beq
  {C \sqrt{k(2n)} k(\nu-n)^n} \oint_{\Gamma} |d z |  | f_\beta(z) e^{-z \tau}| \left( \frac{1}{d(z,\sigma(\Hchli))^{n+1}} +  \frac{1}{d(z,\sigma(\Hchli)) } \right)
\eeq
\neu{where the second term between brackets accounts for the case $y = x$.
Let us now fix the contour $\Gamma$ to be a rectangle: one side is  $\iu [-\epsilon^{1/2},\epsilon^{1/2}] -( \rho+\epsilon/2)$ and the opposite side is $\iu [-\epsilon^{1/2},\epsilon^{1/2}] -(\norm{\Hchli}+\epsilon)$. Then the contour integral is bounded by 
$$
C\E^{-|\tau| (\rho+\epsilon/2)}  \big[\epsilon^{-n} +  (\epsilon^{-(n+1)/2} +\epsilon^{-1/2})\norm{\Hchli}\big].
$$
Since $\epsilon \leq \norm{\Hchli}$ and $C$ is allowed to depend on $\norm{\Hchli}$, the bracketed expression $\big[\ldots\big]$ can be replaced by $\epsilon^{-n}$.and hence we get \eqref{eq: proof reduces}}.

\end{proof}

To obtain Theorem \ref{previewgaptheo}, \neu{we use Theorem \ref{gaptheo} which requires  a bound on $\alpha_\rho$.} To get that bound, note that  for  an exponentially decaying kernel $h$, $k(\nu) $ is bounded uniformly in the volume for $2\nu>d$. Hence we can apply Theorem \ref{thm: decay from gap} with $n$ being the smallest integer greater than $d/2$.

\end{document}